\RequirePackage{amsmath}
\documentclass[runningheads]{llncs}
\usepackage{amsfonts}
\usepackage{color}
\usepackage{pgfplots}
\pgfplotsset{compat=1.14}
\usepackage{xparse}
\usepackage{algorithm}
\usepackage{multirow}
\usepackage[style=base]{caption}
\usepackage{subcaption}
\usepackage{cases}
\usepackage{setspace}
\let\Algorithm\algorithm
\renewcommand\algorithm[1][]{\Algorithm[#1]\setstretch{0.9}}
\usepackage[noend]{algpseudocode}
\usepackage{threeparttablex}
\usepackage{tikz}
\usepackage{hyperref}
\usetikzlibrary{arrows,patterns,positioning}
\usepackage{graphicx}
\usepackage{appendix}
\usepackage{xspace}
\providecommand{\lra}{LRA\xspace}
\providecommand{\rdl}{RDL\xspace}

\providecommand{\Tau}{\mathrm{T}\xspace}
\NewDocumentCommand\Rm{mg}{%
    \IfNoValueTF{#2}
    {\mathbb{R}_{\max}^{#1}}
    {\mathbb{R}_{\max}^{#1 \times #2}}
}
\newcommand{\Rmax}{\mathbb{R}_{\max}}
\newcommand*{\QEDB}{\hfill\ensuremath{\square}}
\newcommand{\tuple}[1]{\langle #1 \rangle}
\newcommand{\CommentX}[1]{\mbox{}\hfill~#1~}

\begin{document}
\title{Computation of the Transient in Max-Plus Linear Systems via SMT-Solving }
\titlerunning{Computation of Transient in MPL Systems via SMT-Solving }
%
\author{
Alessandro Abate\inst{1}
\and
Alessandro Cimatti\inst{2}
\and
Andrea Micheli \inst{2}
\and
Muhammad Syifa'ul Mufid \inst{1}
}
\authorrunning{Alessandro Abate et al.}
%
\institute{Department Computer Science, University of Oxford, UK\\
\email{\{alessandro.abate,muhammad.syifaul.mufid\}@cs.ox.ac.uk\\}
\and
Fondazione Bruno Kessler, Italy\\
\email{\{cimatti,amicheli\}@fbk.eu}
}
\maketitle              
\begin{abstract}


This paper proposes a new approach, grounded in Satisfiability Modulo Theories (SMT), to study the transient of a Max-Plus Linear (MPL) system, that is the number of steps leading to its periodic regime. Differently from state-of-the-art techniques, our approach allows the analysis of periodic behaviors for subsets of initial states, as well as the characterization of sets of initial states exhibiting the same specific periodic behavior and transient. Our experiments show that the proposed technique dramatically outperforms state-of-the-art methods based on max-plus algebra computations for systems of large dimensions. 
\end{abstract}
%
%
%
\setcounter{footnote}{0}
\section{Introduction}
Max-Plus Linear (MPL) systems are a class of discrete-event systems (DES) that are based on the max-plus algebra, an algebraic system using the two operations of maximisation and addition. MPL systems are employed to model applications with features of synchronization without concurrency, and as such are widely used for applications in transportation networks \cite{Baccelli}, manufacturing \cite{Heidergott} and biological systems \cite{Chris,Comet}. In MPL models, the states correspond to time instances related to discrete events. 

A fundamental and well-studied property of MPL systems is related to the periodic behavior of its states: from an initial vector, the trajectories of an MPL system are eventually periodic (in max-plus algebraic sense) starting from a specific event index called the \textit{transient}, and with a specific period called  \textit{cyclicity} \cite{Baccelli}. As explained in \cite[Section 3.1]{Heidergott}, the transient is closely related to the notion of \textit{cycle-time} vector, which governs the asymptotic behaviour of MPL systems. 

The transient is key to solve a number of fundamental problems of MPL systems, such as reachability analysis \cite{Mufid2020} and bounded model checking \cite{Mufid2019}:   
it plays a crucial role as the ``completeness threshold'' (namely, the maximum iteration that is sufficient for the termination of the algorithm) \cite{CT} for those two problems. 
The computation of the transient is an interesting problem, as it is in general not correlated to the dimension of the MPL system. Thus, it is possible for the resulting transient to be relatively large for a small-dimensional MPL system.  
There are
several known upper bounds \cite{Charron,Merlet,Nowak,Gerardo} for the transient, which are mostly computed via the corresponding precedence graph and are, in practice, much larger than the actual values.

This paper has two specific contributions. The first is to provide a novel procedure to compute the transient by means of Satisfiability Modulo Theory (SMT) solving~\cite{SMT}. 
The main idea underpinning the new method is to transform the problem instance into a formula in difference logic, and then passing the formula into an SMT solver, which outputs the transient. 
More precisely, in order to check the validity of the formula, we check the unsatisfiability of its negation. If the SMT solver reports ``satisfied'', then the original formula admits a counterexample, from which we can refine the formula. 
On the other hand, if SMT solver reports ``unsatisfied'', then from the formula we obtain the transient and the corresponding cyclicity. 
The second contribution of this work is to provide a procedure to synthesize the subset of the state space of an MPL system that corresponds to a specific transient/cyclicity pair.  
We show that one can partition the state space into sets corresponding to different  transient/cyclicity pairs.

The rest of the paper is structured as follows. 
Section 2 describes the basics of MPL systems, including the key notion of cycle-time vector. 
In Section 3, we provide the formal definition of transient over MPL systems and also a standard linear algebra procedure, based on matrix multiplication, to compute the transient 
(cf. Algorithm \ref{comp_cyc_trans}), 
which is later used as a benchmark. 
Section 4 is divided into four parts. The first part provides the background on SMT and including the underlying relevant theory. The translation of inequalities over max-plus algebra to formulae in difference logic is explained in the second part. In the third part, we provide SMT-based methods (cf. Algorithms \ref{trans_cone_smt} and \ref{comp_cyc_trans_smt}) 
to compute the transient. The spatial synthesis problem is discussed in the last part. 
The comparison of the performance of the novel algorithm against the standard linear algebra procedure is presented in Section 5. The paper is concluded with Section 6. The developed code and generated data can be found in \url{https://es.fbk.eu/people/amicheli/resources/formats20/}.

\section{Preliminaries}

\subsection{Max-Plus Linear Systems}

Max-plus algebra is a modification of linear algebra derived over the max-plus semiring  $(\Rmax,\oplus,\otimes)$ where
$\Rmax:=\mathbb{R}\cup\{\varepsilon:=-\infty\}$ and 
$a\oplus b := \max\{a,b\}, a\otimes b := a+b$,
for all $a,b\in \Rmax$. The zero and unit elements of $\Rmax$ are $\varepsilon$ and 0, respectively. 
The max-plus algebraic operations can be extended to matrices and vectors in a natural way. For $A,B\in \Rm{m}{n}, C\in \Rm{n}{p}$ and $\alpha\in \Rmax$,
\begin{align}
\nonumber [\alpha\otimes A](i,j)&=\alpha + A(i,j),\\
\nonumber [A\oplus B](i,j)&=A(i,j) \oplus B(i,j),\\
\nonumber [A\otimes C](i,j)&=\bigoplus_{k=1}^n A(i,k)\otimes C(k,j).
\end{align}
Given $A\in\Rm{n}{n}$ and $t\in \mathbb{N}$, $A^{\otimes t}$ denotes $A\otimes\ldots\otimes A$ ($t$ times). For $t=0$, $A^{\otimes 0}$ is an $n$-dimensional max-plus identity matrix where all diagonal and non-diagonal elements are $0$ and $\varepsilon$, respectively.

Given $V=\{v_1,\ldots,v_p\}$ as a set of vectors in $\Rmax^n$, we use the same notation to denote a matrix where all columns are in $V$ i.e., $V(\cdot,i)=v_i$ for $1\leq i \leq p$. A vector $v\in \mathbb{R}^n$
is \textit{a max-plus linear combination} of $V$ if $v=\alpha_1\otimes v_1\oplus \ldots \oplus \alpha_p \otimes v_p$ for some scalars $\alpha_1,\ldots,\alpha_p\in \mathbb{R}$
or equivalently there exists $w\in \mathbb{R}^p$ such that $V\otimes w=v$. The set of all max-plus linear combinations of $V$ is called \textit{max-plus cone}\footnote{Unlike in \cite{Butkovic,Gaubert}, we require each max-plus cone to be a subset of $\mathbb{R}^n$.} and is denoted by $\mathsf{cone}(V)$ \cite{Butkovic}. It is formally expressed as 
\begin{equation}
\label{cone}
    \mathsf{cone}(V)=\{V\otimes w\mid w \in \mathbb{R}^p\}.
\end{equation}
Furthermore, we denote as $v_1,\ldots,v_p$ the basis of  $\mathsf{cone}(V)$. Notice that the max-plus cone is closed under the  operations $\oplus$ and $\otimes$: if $v,w$ are in $\mathsf{cone}(V)$, then so is $\alpha\otimes v\oplus \beta \otimes w$ for $\alpha,\beta\in \mathbb{R}$. Max-plus cones are the analogues of vector subspaces in classical linear algebra.

A dynamical system over the max-plus algebra is called a Max-Plus Linear (MPL) system and is defined as 
\begin{equation}
\textbf{x}(k+1)=A\otimes \textbf{x}(k), ~~k=0,1,\ldots
\label{mpl}
\end{equation}
where $A\in \Rm{n}{n}$ is the system matrix, and vector $\textbf{x}(k)=[x_1(k)~\ldots~ x_n(k)]^\top$ encodes the state variables \cite{Baccelli}. For example, $\textbf{x}$ can be used to represent the time stamps associated to the discrete events, while $k$ corresponds to the events counter. Hence, it is more convenient to consider $\mathbb{R}^n$ (instead of $\Rmax^n)$ as the state space. Applications of MPL systems are significantly found on systems where the time variable is essential, such as transportation networks \cite{Heidergott}, scheduling or  \cite{Alirezaei} manufacturing  \cite{Aleksey} problems, or biological systems \cite{Chris,Comet}. 


\begin{definition}[Precedence Graph \cite{Baccelli}]
\normalfont
The precedence graph of $A\in\Rm{n}{n}$, denoted by $\mathcal{G}(A)$, is a weighted directed graph with nodes $1,\ldots,n$ and an edge from $j$ to $i$ with weight $A(i,j)$ for each $A(i,j)\neq \varepsilon$.\QEDB
\end{definition}

\begin{definition}[Regular Matrix \cite{Heidergott}]
\normalfont
A matrix $A\in \Rm{n}{n}$ is called \textit{regular} if $A$ contains at least one finite element in each row. \QEDB
\end{definition}

\begin{definition}[Irreducible Matrix \cite{Baccelli}]
\normalfont
A matrix $A\in \Rm{n}{n}$ is called \textit{irreducible} if $\mathcal{G}(A)$ is strongly connected.\QEDB
\end{definition}

Recall that a directed graph is strongly connected if, for any two different nodes $i,j$, there exists a path from $i$ to $j$. The weight of a path $p=i_1i_2\ldots i_k$ is equal to the sum of the edge weights in $p$. A circuit, namely a path that begins and ends at the same node, is called \textit{critical} if it has maximum average weight, which is the weight divided by the length of the path \cite{Baccelli}.

Each irreducible matrix $A\in\Rmax^{n\times n}$ admits a unique max-plus eigenvalue $\lambda\in\mathbb{R}$ and a corresponding max-plus eigenspace $E(A)=\{{x}\in\mathbb{R}^n\mid A\otimes {x}=\lambda\otimes {x}\}$\footnote{Because we regard  $\mathbb{R}^n$ to be the state space of the MPL system \eqref{mpl}, we only consider eigenvectors with finite  elements.}. The scalar $\lambda$ is equal to the average weight of critical circuits in $\mathcal{G}(A)$, and $E(A)$ can be computed from $A_\lambda^+=\bigoplus_{k=1}^n((-\lambda)\otimes A)^{\otimes k}$. More specifically, $E(A)$ is the max-plus linear combination of the $i^{th}$ column of $A_\lambda^+$, for $i$ such that $A_\lambda^+(i,i)=0$ \cite{Baccelli}. Thus, the eigenspace $E(A)$ is a max-plus cone. A reducible matrix may have multiple eigenvalues, where the maximum one equals to the average weight of critical circuits of $\mathcal{G}(A)$.
\begin{example}
Consider a two-dimensional MPL system $\textbf{x}(k+1)=A\otimes \textbf{x}(k)$,  with 
\[
A=\begin{bmatrix}
    2 &~5 \\
    3&~3
\end{bmatrix}, 
\]
which represents a simple railway network between two cities \cite[Sec. 0.1]{Baccelli}, as shown in \autoref{railway}. The dynamics w.r.t. \eqref{mpl} can be expressed as
\[
\begin{bmatrix}
    x_1(k+1)\\x_2(k+1)
\end{bmatrix}=
\begin{bmatrix}
    \max\{x_1(k)+2,x_2(k)+5\}\\
    \max\{x_1(k)+3,x_2(k)+3\}
\end{bmatrix}.
\]
For $1\leq i,j\leq 2$, the element $A(i,j)$ corresponds to the time taken to travel from station $S_j$ to $S_i$, while $x_i(k)$ is the time of the $k$-th departure at station $S_i$.
\begin{figure}[!ht]
\centering
\includegraphics{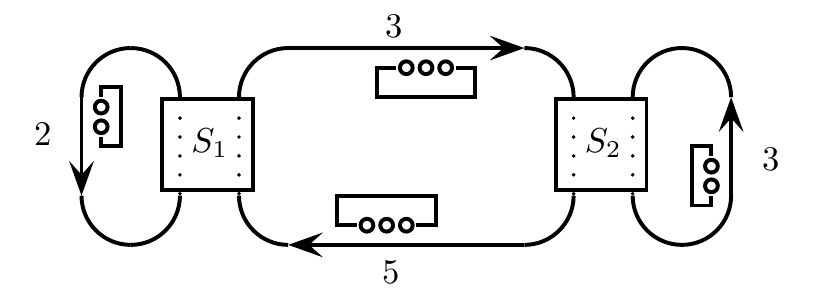}
\caption{A simple railway network represented by an MPL system.}
\label{railway}
\end{figure}\\
From an initial vector, say  $\textbf{x}(0)=[0~0]^\top$, one can compute vectors denoting the next departure times,  as follows
\[
\begin{bmatrix}
    5\\3
\end{bmatrix},
\begin{bmatrix}
    8\\8
\end{bmatrix},
\begin{bmatrix}
    13\\11
\end{bmatrix},
\begin{bmatrix}
    16\\16
\end{bmatrix},\ldots.
\]
Leaving the details aside, the matrix $A$ has eigenvalue $\lambda=4$ and eigenspace $E(A)=\{\textbf{x}\in \mathbb{R}^2\mid x_1-x_2=1\}$. \QEDB
\end{example}
\subsection{Cycle-Time Vector}
This section presents the definition of cycle-time vector of MPL systems. The computation of the cycle-time vector is indeed important, as it can shed light on the asymptotic behavior of MPL systems. In this section, we show its relationship with the eigenspace and eigenvalue of the underlying state matrix. Furthermore, as it will be clear in Section 3, the cycle-time vector can be used to determine whether the states of an MPL system are  eventually periodic.
\begin{definition}[Cycle-Time Vector \cite{Heidergott}]
\normalfont
\label{cycle_time}
Consider a regular MPL system \eqref{mpl}, and assume that for all $j\in \{1,\ldots,n\}$ the quantity $\eta_j$, defined by 
\[\displaystyle \eta_j=\lim_{k\rightarrow +\infty} (x_j(k)/k),\]
exists. Then the vector $\chi=[\eta_1~\ldots~\eta_n]^\top$ is called the the cycle-time vector of the given sequence $\textbf{x}(k)$ with respect to $A$.\QEDB
\end{definition}

It has been shown in \cite[Theorem 3.11]{Heidergott} that if the cycle-time vector of $A$ exists for at least one initial vector then it exists for any initial vector. Instead of computing the limit as in \autoref{cycle_time}, the cycle-time vector can be generated using a procedure \cite[Algorithm 31]{Fahim}.

\begin{theorem}[\cite{Fahim}]
\normalfont
\label{cycle_time_transient}
Suppose we have a regular MPL system \eqref{mpl}. For each $\textbf{x}(0)\in \mathbb{R}^n$ there exist natural numbers $p,q$ such that $\textbf{x}(k+q)=(q\times \chi) + \textbf{x}(k)$ for all $k\geq p$, where $\chi=[\eta_1~\ldots~\eta_n]^\top$ is the cycle-time vector of $A$ and the multiplication $q\times \chi$ is defined in the classical algebra. \QEDB
\end{theorem}

By \autoref{cycle_time_transient}, the trajectories of a regular MPL system \eqref{mpl} starting from any initial vector is governed by the corresponding cycle-time vector $\chi$. In general, the elements of $\chi$ may be different, as shown in \cite[Example 1]{Fahim}. However, if $E(A)\neq \emptyset$ then the elements of $\chi$ are all equal.
\begin{proposition}
\label{cycle_time_unique}
\normalfont
Suppose a regular MPL system \eqref{mpl} has maximum eigenvalue $\lambda$. The eigenspace $E(A)$ is not empty iff $\chi=[\lambda~\ldots~\lambda]^T\in \mathbb{R}^n$.
\end{proposition}
\begin{proof}
$(\Rightarrow)$ Suppose $E(A)\neq \emptyset$. By taking $\textbf{x}(0)\in E(A)$, we have $\textbf{x}(k+1)=\lambda\otimes \textbf{x}(k)$ for $k \geq 0$  which implies $x_j(k)=(\lambda \times k) + x_j(0)$ for $j\in \{1,\ldots,n\}$. It is straightforward that $\lim_{k\rightarrow +\infty} x_j(k)/k=\lambda$ for all $j\in\{1,\ldots,n\}$. \\
$(\Leftarrow)$ Suppose $\textbf{x}(0)\in\mathbb{R}^n$. By \autoref{cycle_time_transient}, there exist $p,q$ such that $\textbf{x}(p+q)=q\times \chi + \textbf{x}(p)$. Because $\chi=[\lambda~\ldots~\lambda]$, it can be written as $\textbf{x}(p+q)=(q\times \lambda)\otimes \textbf{x}(p)$. Let $$v=\bigoplus_{i=1}^{q}(\lambda\times  (q-i))\otimes \textbf{x}(p+i-1).$$
One could check that $v\in \mathbb{R}^n$ and $A\otimes v=\lambda\otimes v$. Thus, $E(A)\neq \emptyset$.\QEDB 
\end{proof}
\section{Transient in Max-Plus-Linear Systems}
The transient of MPL systems is related to the sequence of the powers of matrix $A$, namely $A^{\otimes{k}}$ for $k\geq 0$.
\begin{proposition}[Transient \cite{Baccelli,Heidergott}] 
\label{transient_bound}
\normalfont
For an irreducible matrix $A\in\Rmax^{n\times n}$ and its max-plus eigenvalue $\lambda\in\mathbb{R}$, there exist $k_0,c\in\mathbb{N}_0$, such that $A^{\otimes(k+c)}=(\lambda\times c)\otimes A^{\otimes k}$ for all $k\geq k_0$. The smallest such $k_0$ and $c$ are called the \textit{transient} and the \textit{cyclicity} of $A$, respectively. \QEDB
\label{trans}
\end{proposition}

For the rest of this paper, we denote the transient and the cyclicity of $A$ as $\mathsf{tr}(A)$ and $\mathsf{cyc}(A)$, respectively. While $\mathsf{cyc}(A)$ is related to critical circuits in the precedence graph $\mathcal{G}(A)$ (see \cite[Definition 3.94]{Baccelli} for more details\footnote{In this reference, one can find the cyclicity for reducible and irreducible matrices using graph-theoretical approaches.}), $\mathsf{tr}(A)$ is unrelated to the dimension of $A$. Even for a small $n$, the transient of $A\in \Rm{n}{n}$ can be large. Upper bounds of the transient have been discussed in \cite{Charron,Merlet,Nowak,Gerardo}. 

By \autoref{trans}, each \emph{irreducible} MPL system enjoys a \textit{periodic} behaviour with a rate $\lambda$: for each initial vector $\textbf{x}(0)\in \mathbb{R}^n$ we have $\textbf{x}(k+\mathsf{cyc}(A))=(\lambda \times \mathsf{cyc}(A))\otimes \textbf{x}(k)$ for all $k\geq \mathsf{tr}(A)$ where the vectors $\textbf{x}(1),\textbf{x}(2)$ are computed recursively by \eqref{mpl}. A similar condition may be found on 
reducible MPL systems: we denote the corresponding transient and cyclicity as global, as per \autoref{transient_bound}. The local transient and cyclicity for a specific initial vector $\textbf{x}\in\mathbb{R}^n$ and for a set $X\subseteq \mathbb{R}^n$ is defined as follows. 
\begin{definition}
\normalfont
Given $A\in\Rm{n}{n}$ with maximum eigenvalue $\lambda$ and an initial vector $\textbf{x}\in \mathbb{R}^n$, the local transient and cyclicity of $\textbf{x}(0)$ w.r.t. $A$ are respectively the smallest $k_0,c\in \mathbb{N}_0$ such that
$\textbf{x}(k+c)=\lambda c\otimes \textbf{x}(k)$ for all $k\geq k_0$.
We denote those scalars as $\mathsf{tr}(A,\textbf{x})$ and $\mathsf{cyc}(A,\textbf{x})$, respectively. Furthermore, for $X\subseteq \mathbb{R}^n,\mathsf{tr}(A,X)=\max\{\mathsf{tr}(A,\textbf{x}(0))\mid \textbf{x}(0)\in X\}$ and $\mathsf{cyc}(A,X)=\texttt{lcm}\{\mathsf{cyc}(A,\textbf{x}(0))\mid \textbf{x}(0)\in X\}$, where $\texttt{lcm}$ stands for the ``least common multiple''. \QEDB
\label{trans_set}
\end{definition}

By definition, we have $\mathsf{tr}(A,\mathbb{R}^n)=\mathsf{tr}(A)$. For a max-plus cone $X=\mathsf{cone}(V)$, we show that the local cyclicity and transient can be computed from the corresponding bases, provided that $\mathsf{tr}(A,v_i)$ exists for all $1\leq i\leq p$.
\begin{proposition}
\normalfont
Given a max-plus cone $X=\mathsf{cone}(V)$ where $V=\{v_1,\ldots,v_p\}$, we have $\mathsf{tr}(A,X)=\mathsf{tr}(A,V)=\max\{\mathsf{tr}(A,v)\mid v\in V \}$, and $\mathsf{cyc}(A,X)=\mathsf{cyc}(A,V)=\texttt{lcm}\{\mathsf{cyc}(A,v) \mid  v\in V\}$. 
\end{proposition}
\begin{proof}
Suppose $\textbf{x} \in X$. It follows that there exist scalars $\alpha_1,\dots,\alpha_p$ such that
$\textbf{x} = \bigoplus_{i=1}^p (\alpha_i \otimes v_i)$.
Let $k_0^\ast = \max\{\mathsf{tr}(A,v_i)\mid 1\leq i\leq p \}$ and $c^\ast = \texttt{lcm}\{\mathsf{cyc}(A,v_i) \mid  1\leq i\leq p\}$.
Then, we obtain
\begin{align*}
A^{\otimes (k_0^\ast + c^\ast)} \otimes \bigoplus_{i=1}^p (\alpha_i \otimes v_i)&= \bigoplus_{i=1}^p \left(\alpha_i \otimes A^{\otimes (k_0^\ast + c^\ast)} \otimes v_i\right) , \\
&= \bigoplus_{i=1}^p \left(\alpha_i \otimes \lambda^{\otimes c^\ast} \otimes A^{\otimes k_0^\ast} \otimes v_i\right) , \\
&= \lambda^{\otimes c^\ast} \otimes A^{\otimes k_0^\ast} \otimes \bigoplus_{i=1}^p \left(\alpha_i  \otimes v_i\right) , 
\end{align*}
which shows that $\mathsf{tr}(A,\textbf{x}) \leq k_0^\ast$ and $\mathsf{cyc}(A,\textbf{x}) \leq c^\ast$. Thus, $\mathsf{tr}(A,X) \leq \mathsf{tr}(A,V)$ and $\mathsf{cyc}(A,X) \leq \mathsf{cyc}(A,V)$. On the other hand, because $V\subseteq X$, we have $\mathsf{tr}(A,V)\leq \mathsf{tr}(A,X)$ and $\mathsf{cyc}(A,V)\leq \mathsf{cyc}(A,X)$. Hence, we can conclude that $\mathsf{tr}(A,V)= \mathsf{tr}(A,X)$ and $\mathsf{cyc}(A,V)= \mathsf{cyc}(A,X)$.\QEDB
\end{proof}
\begin{definition}
\normalfont
Suppose we have a regular matrix $A\in \Rm{n}{n}$. The underlying MPL system \eqref{mpl} is classified into three categories as follows:
\begin{itemize}
    \item[i.] \textit{never periodic}: $\mathsf{tr}(A,\textbf{x}(0))$ does not exist for all $\textbf{x}(0)\in \mathbb{R}^n$,
    \item[ii.] \textit{boundedly periodic}: $\mathsf{tr}(A,\textbf{x}(0))$ exists for all $\textbf{x}(0)\in \mathbb{R}^n$ and $\mathsf{tr}(A)$ exists,
    \item[iii.] \textit{unboundedly periodic}: $\mathsf{tr}(A,\textbf{x}(0))$ exists for all $\textbf{x}(0)\in \mathbb{R}^n$ but $\mathsf{tr}(A)$ does not.
\end{itemize}
We call \eqref{mpl} \textit{periodic} if it is either \textit{unboundedly periodic} or \textit{boundedly periodic}. \QEDB
\end{definition}

We show that the periodic behavior of an MPL system is indeed related to the eigenspace and cycle-time vector of its corresponding state matrix.

\begin{theorem}
\label{periodic_everywhere}
\normalfont
Suppose we have a regular matrix $A\in \Rm{n}{n}$ with a maximum eigenvalue $\lambda$ and cycle-time vector $\chi$. The following statements are equivalent.
\begin{itemize}
    \item[$a$.] The underlying MPL system \eqref{mpl} is \textit{periodic}.
    \item[$b$.] The corresponding cycle-time vector is $\chi=[\lambda~\ldots~\lambda]^\top\in \mathbb{R}^n$.
    \item[$c$.] The eigenspace $E(A)$ is not empty. 
\end{itemize}
\end{theorem}
\begin{proof}
By Proposition \ref{cycle_time_unique}, it is suffice to prove $((a)\Rightarrow (c))$ and $(b)\Rightarrow (a)$.\\
$(a)\Rightarrow (c)$. Suppose $\textbf{x}(0)\in\mathbb{R}^n$. As $A$ is periodic, there exist natual numbers $k_0,c$ such that $\textbf{x}(k+c)=\lambda c\otimes \textbf{x}(k)$ for all $k\geq k_0$. One could check that
$$v=\bigoplus_{i=1}^{c}(\lambda\times  (c-i))\otimes \textbf{x}(k_0+i-1).$$
is an eigenvector of $A$. $((b)\Rightarrow (a))$. By Theorem \ref{cycle_time_transient}, for each $\textbf{x}(0)\in \mathbb{R}^n$ there exist natural numbers $p,q$ such that $ \textbf{x}(k+q)=(q\times \chi)+ \textbf{x}(k)$ for all $k\geq p$. Because $\chi=[\lambda~\ldots~\lambda]$, it can be written as $ \textbf{x}(k+q)=(q\times \lambda)\otimes \textbf{x}(k)$. This shows that $\mathsf{tr}(A,\textbf{x}(0))$ exists for all $\textbf{x}(0)\in\mathbb{R}^n$. Therefore, \eqref{mpl} is \textit{periodic}. \QEDB
\end{proof}
\begin{proposition}
\normalfont
\label{unbounded_periodic}
Suppose we have a regular matrix $A\in \Rm{n}{n}$ with maximum eigenvalue $\lambda$ and non-empty eigenspace $E(A)$. If there exist $i\in \{1,\ldots,n\}$ and natural numbers $k_0^\prime,c^\prime$ such that $A^{\otimes k+c^\prime}(\cdot,i)=\mu c^\prime\otimes A^{\otimes k}(\cdot,i)$ for all $k\geq k_0^\prime$ with $\mu<\lambda$ then \eqref{mpl} is \textit{unboundedly periodic}. 
\end{proposition} 
\begin{proof}
Because $E(A)\neq \emptyset$, by Theorem \ref{periodic_everywhere}, the corresponding MPL system \eqref{mpl} is \textit{periodic}. However, as $A^{\otimes k+c^\prime}(\cdot,i)=\mu c^\prime\otimes A^{\otimes k}(\cdot,i)$ where $\mu<\lambda$ for all $k\geq k_0^\prime$, it is deemed impossible to find $k_0,c$ such that $A^{\otimes k+c}=\lambda c\otimes A^{\otimes k}$ for $k\geq k_0$. Consequently, \eqref{mpl} is \textit{unboundedly} periodic. \QEDB
\end{proof}

We now will provide the procedure to compute the transient of MPL systems. As per \autoref{transient_bound}, the common method to obtain the (global) transient of $A\in\Rm{n}{n}$ is by computing the power of the matrix $A^{\otimes 0},A^{\otimes 1},\ldots$ until we find $k_0\geq 0$ such that $A^{\otimes(k_0+c)}=\lambda^{\otimes c}\otimes A^{\otimes k_0}$ where $\lambda,c$ is respectively the max-plus eigenvalue and cyclicity of $A$. Similarly, to find the transient of $A$ w.r.t. a max-plus cone $X=\mathsf{cone}(V)$ one needs to compute $A^{\otimes 0}\otimes V,A^{\otimes 1}\otimes V,\ldots$.

Algorithm \ref{comp_cyc_trans} illustrates the procedure to compute transient (and cyclicity) for a max-plus cone $\mathsf{cone}(V)$ w.r.t. $A\in \Rm{n}{n}$. While originally designed for irreducible matrices, it also can be applied to find the transient of reducible matrices (if any). For this reason, we assign a maximum bound as termination condition. It is important to note that Algorithm \ref{comp_cyc_trans} can also be used to compute the local transient and cyclicity for a vector: that is, when $V$ has only one column. The algorithm starts by computing the cycle-time vector $\chi$ of the state matrix. If the entries of $\chi$ are not all the same then the transient for $\mathsf{cone}(V)$ does not exist. In line 11, we perform equality checking w.r.t. a scalar between $A^{\otimes {it-m}}\otimes V$ and $A^{\otimes {it}}\otimes V$. 

By Theorem \ref{periodic_everywhere} and Proposition \ref{unbounded_periodic}, one can classify an MPL system \eqref{mpl} into a category in Definition 6. As a result, determining the existence of global transient is a decidable problem. For \textit{boundedly periodic} MPL systems, computing the global transient is also a decidable problem. This is because they ensure the existence of a finite transient, meaning that Algorithm \ref{comp_cyc_trans} eventually terminates. However, Algorithm \ref{comp_cyc_trans} is sound but does not necessarily terminate (in general) for \textit{unboundedly periodic} MPL systems. 
\begin{algorithm}[H]
	\footnotesize
    \caption{\footnotesize Computation of cyclicity and transient of $A$ w.r.t. $\mathsf{cone}(V)$}
    \label{comp_cyc_trans}
    \centering
    \begin{algorithmic}[1]
\Function{TransCone}{$A,V,N$}
\State $\textbf{M}\gets\Call{EmptyVector}$ \Comment{ empty vector used to store}
\State $\textbf{M}.\mathsf{push}\_\mathsf{back}(V)$ \CommentX{$A^{0}\otimes V,A^{1}\otimes V,\ldots$\hspace{10.0ex}}
\State $it\gets 0$\Comment{number of iterations}
\State $\chi\gets \Call{CycleTimeVector}{A}$ \Comment{computing cycle-time vector}
\If{elements of $\chi$ are all equal}
\While{($it\leq N)$}
\State $\textbf{M}.\mathsf{push}\_\mathsf{back}(A\otimes \textbf{M}[it])$
\State $it\gets it+1$
\For{$1\leq m<it $}
\If{$(\textbf{M}[it]=(\lambda\times m)\otimes \textbf{M}[it-m])$}
\State{\Return{$\tuple{it-m, m}$}}
\EndIf
\EndFor
\EndWhile
\If{$(it> N)$}
\State \textbf{print} ``terminated after reaching maximum bound''
\EndIf
\Else
\State \textbf{print} ``the transient does not exist''
\EndIf
\EndFunction
    \end{algorithmic}
\end{algorithm}\vspace*{-2.5ex}
\begin{remark}
The procedure in Algorithm \ref{comp_cyc_trans} only employs matrix operations in max-plus algebra. It can be improved by computing the cyclicity of the matrix from the corresponding precedence graph. If the resulting cyclicity is $c$ then the range in line 10 of Algorithm \ref{comp_cyc_trans} can be taken between $1$ and $c$. \QEDB
\end{remark}
\begin{example}
\label{ex_trans}
Suppose we have a regular and reducible MPL system $\textbf{x}(k+1)=B\otimes  \textbf{x}(k)$, where
\begin{equation}
\label{reducible_matrix}
B=\begin{bmatrix}
2 &~8&~\varepsilon\\
10&~5&~\varepsilon\\
3&~\varepsilon&~a
\end{bmatrix},
\end{equation}
and where $a\geq 8$. The corresponding eigenvalue for $B$ is $\lambda=9$ if $8\leq a\leq 9$; $\lambda=a$ otherwise. Taking the power of the matrix, we have 
\[
B^{\otimes 2}=
\begin{bmatrix}
    18  &  13 & \varepsilon\\
    15   & 18 &\varepsilon\\
    a+3  &  11 &   2a
\end{bmatrix},
B^{\otimes 3}=
\begin{bmatrix}
23 &   26 & \varepsilon\\
    28 &   23 & \varepsilon\\
    b&   a+11   & 3a
\end{bmatrix},
B^{\otimes 4}=
\begin{bmatrix}
    36 &  31& \varepsilon\\
    33 &  36& \varepsilon\\
    a+b  &  2a+11 &  4a
\end{bmatrix},
\]
where $b=\max\{21,2a+3\}$. One can check that, for $a>9$, the matrix does not admit an eigenvector over $\mathbb{R}^3$ (but it still has eigenvector over $ \Rm{3}$). As a result, $B$ is \textit{never periodic}. 

On the other hand, for $8\leq a\leq 9$, the corresponding $E(B)$ is not empty. Thus, $B$ is \textit{periodic}. Furthermore, for $k\geq 2$, we have
$$[B^{\otimes k+2}](i,\cdot)=\left\{
\begin{array}{ll}
18\otimes[B^{\otimes k}](\cdot,i),&~\text{if}~i\in \{1,2\},\\
2a\otimes[B^{\otimes k}](\cdot, i),&~\text{if}~i=3,
\end{array}
\right.
$$
which shows that $B$ is \textit{boundedly periodic} with global transient $\mathsf{tr}(B)=2$ if and only if $a=9$. Thus, when $8\leq a<9$ $B$ is \textit{unboundedly periodic}.
\QEDB
\end{example}

\section{Computation of Transient of MPL Systems with SMT}

This section describes a new procedure to compute the transient of MPL systems by means of Satisfiability Modulo Theories (SMT). We first mention some basic notions on SMT.

\subsection{Background on SMT}

Given a first-order formula $\psi$ in a background theory $\Tau$, the Satisfiability Modulo Theory (SMT) problem consists in deciding
whether there exists a model (i.e.\ an assignment to the free variables in $\psi$) that satisfies $\psi$ \cite{SMT}. For example, consider the formula $ (x \le y) \wedge (x + 3 = z) \vee (z \ge y) $ within the theory of real numbers. The formula is satisfiable and a valid model is $\{x:=5, \: y:=6, \: z:=8\}$. 

SMT solvers can support different theories. A widely used theory is Linear Real Arithmetic (\lra). A formula in \lra is an arbitrary Boolean combination, or universal ($\forall$) and existential ($\exists$) quantification, of atoms in the form $ \sum_i a_i x_i \bowtie c$ where $\bowtie \in \{>, <, \ge, \le, \ne, =\}$, every $x_i$ is a real variable, and every $a_i$ and $c$ are rational constants. Difference logic (\rdl) is the subset of \lra in which all atoms are restricted to the form $x_i - x_j \bowtie c$. Both theories are decidable \cite[Section 26.2.2.2]{SMT}.

\subsection{From Max-Plus Algebra to Difference Logic}

Before providing the main contribution, we show that the inequalities in max-plus algebra can be expressed as a formula in difference logic. For the rest of this paper, $\sim$ is either $\geq $ or $>$. We write $\neg(a \sim b)$ if it is not the case that $a\sim b$.

\begin{proposition}
\normalfont
\label{ineq_maxplus}
Given $a_1,\ldots,a_p,a, b\in \Rmax$, real-valued variables $\texttt{x}_1,\ldots,\texttt{x}_p $,  and $1\leq j\leq p$,  we have 
\begin{eqnarray}
\bigoplus_{i=1}^p (\texttt{x}_i+a_i)\sim a &\equiv&\bigvee_{i=1}^p (\texttt{x}_i+a_i\sim a),\label{GE}\\
a \sim \bigoplus_{i=1}^p(\texttt{x}_i+a_i) &\equiv& \bigwedge_{i=1}^p (a\sim \texttt{x}_i+a_i),\label{LE}
\end{eqnarray}
\begin{numcases}{\bigoplus_{i=1}^p (\texttt{x}_i+a_i)\sim \texttt{x}_j+b\equiv}
\mathtt{true}, & $\mathrm{if}~(a_j\sim b)$, \label{GE_true}\\
\bigvee_{\substack{i=1\\i\neq j}}^p (\texttt{x}_i+a_i\sim \texttt{x}_j+b), &  otherwise, \label{GE_reduces}
\end{numcases}
\begin{numcases}{\texttt{x}_j+b\sim \bigoplus_{i=1}^p(\texttt{x}_i+a_i)\equiv}
\bigwedge_{\substack{i= 1\\i\neq j}}^p (\texttt{x}_j+b \sim \texttt{x}_i+a_i), & $\mathrm{if}~(b\sim a_j)$, \label{LE_reduce}\\
\mathtt{false}, &  otherwise. \label{LE_false}
\end{numcases}
\end{proposition}
\begin{proof}
The equation $\bigoplus_{i=1}^p (\texttt{x}_i+a_i)\sim a$ is satisfied iff there exists $1\leq i\leq p$ such that $\texttt{x}_i+a_i\sim a$. Similarly, $a \sim \bigoplus_{i=1}^p(\texttt{x}_i+a_i)$ holds iff $a\sim \texttt{x}_i+a_i$ for all $1\leq i\leq p$. Hence, we get \eqref{GE} and \eqref{LE}. 
By applying \eqref{GE}, we have $$\bigoplus_{i=1}^p (\texttt{x}_i+a_i)\sim  \texttt{x}_j+b\equiv\bigvee_{i=1}^p (\texttt{x}_i+a_i\sim \texttt{x}_j+b)\equiv (a_j\sim b)\vee \bigvee_{\substack{i=1\\i\neq j}}^p (\texttt{x}_i+a_i\sim \texttt{x}_j+b).$$
If it is true that $a_j\sim b$ then we get \eqref{GE_true}, otherwise we get \eqref{GE_reduces}. The proof for \eqref{LE_reduce}-\eqref{LE_false} is similar to that of \eqref{GE_true}-\eqref{GE_reduces}.\QEDB
\end{proof}
\begin{proposition}[Reduced Formula]
\normalfont
\label{ori_to_reduced}
Given real valued variables $\texttt{x}_1,\ldots,\texttt{x}_p$ and $a_1,\ldots,a_p,b_1,\ldots,b_p\in \Rmax$, the inequality
\begin{equation}
F\equiv\bigoplus_{i=1}^p(\texttt{x}_i+a_i)\sim \bigoplus_{j= 1}^p (\texttt{x}_j+b_j)
\label{origin}
\end{equation} 
is equivalent to 
\begin{equation}
F^\ast\equiv\bigoplus_{i\in S_1}(\texttt{x}_i+a_i)\sim \bigoplus_{j\in S_2} (\texttt{x}_j+b_j),
\label{reduced}
\end{equation} 
where $S_1= \{1,\ldots,p\}\setminus\{1\leq k\leq p \mid  a_k= \varepsilon~\text{or}~ \neg (a_k\sim b_k)\}$ and $S_2=\{1,\ldots,p\} \setminus\{1\leq k\leq p \mid  b_k= \varepsilon~\text{or}~ a_k\sim b_k\}$, respectively. 
\end{proposition}
\begin{proof}
Suppose we set initially $S_1=S_2=\{1,\ldots,p\}$. By applying \eqref{GE}, \eqref{origin} can be expressed as
\begin{align*}
F\equiv \bigvee_{i\in S_1}\left(\texttt{x}_i+a_i \sim \bigoplus_{j\in S_2} (\texttt{x}_j+b_j)\right). 
\end{align*}
Indeed we can ignore any scalar $k$ if $a_k=\varepsilon$. Furthermore, for each $l$ such that $\neg(a_l\sim b_l)$, by \eqref{LE_false}, we have
\[F\equiv \mathtt{false}~ \vee\!\!\!\! \bigvee_{i\in S_1-\{l\}}\!\!\!\left(\texttt{x}_i+a_i \!\sim\! \bigoplus_{j\in S_2} \!\!(\texttt{x}_j+b_j)\!\right)\!\equiv\! \!\!\!\!\bigvee_{i\in S_1-\{l\}}\left(\texttt{x}_i+a_i \!\sim\! \bigoplus_{j\in S_2} (\texttt{x}_j+b_j)\right).\]
This shows that $l$ can be removed from $S_1$. Similarly, by \eqref{LE}, we have 
\begin{align*}
F\equiv \bigwedge_{j\in S_2}\left(\bigoplus_{i\in S_1} (\texttt{x}_i+a_i)\sim \texttt{x}_j+a_j \right).  
\end{align*}
Again we can ignore any scalar $k$ when $b_k=\varepsilon$. By \eqref{GE_true}, for each $l\in S_2$ such that $a_l\sim b_l$ we have 
$$F\equiv \mathtt{true}~ \wedge\!\!\!\! \bigwedge_{j\in S_2-\{l\}}\!\!\left(\bigoplus_{i\in S_1} (\texttt{x}_i+a_i)\!\sim\! \texttt{x}_j+a_j \!\right)\!\equiv \!\!\!\!\!\bigwedge_{j\in S_2-\{l\}}\!\left(\bigoplus_{i\in S_1} (\texttt{x}_i+a_i)\!\sim\! \texttt{x}_j+a_j \!\right). $$
Hence $l$ can be removed from $S_2$.\QEDB
\end{proof}

Proposition \ref{ori_to_reduced} ensures that any inequality expression in max-plus algebra can be reduced to a simpler one in which no a variable appears on both sides i.e., $S_1 \cap S_2  =\emptyset$. However, $S_1$ and $S_2$ cannot be both empty if there exists at least one finite scalar in both sides of \eqref{origin}. 
We call \eqref{reduced} as a non-trivial reduced formula if both $S_1\neq \emptyset$ and $S_2\neq\emptyset$.
\begin{proposition}
\normalfont
\label{reduced_formula}
Given a non-trivial reduced formula in \eqref{reduced},  then
\begin{equation}
F^\ast\equiv \bigwedge_{j\in S_2 }\left( \bigvee_{i\in S_1 } (\texttt{x}_i-\texttt{x}_j\sim b_j-a_i)\right) \equiv
\bigvee_{i\in S_1 }\left( \bigwedge_{j\in S_2 } (\texttt{x}_i-\texttt{x}_j\sim b_j-a_i)\right).
\end{equation}
If $S_1=\emptyset$ then $F^\ast\equiv \mathtt{false}$. On the other hand, if $S_2=\emptyset$ then $F^\ast\equiv \mathtt{true}$.
\end{proposition}
\begin{proof}
The proof follows directly by applying Proposition \ref{ineq_maxplus} on \eqref{reduced}.\QEDB
\end{proof}

Proposition \ref{reduced_formula} shows that any non-trivial formula of \eqref{reduced} can be expressed as a difference logic formula in disjunctive and conjunctive normal forms. 

\subsection{Procedure to Compute Transient of MPL Systems with SMT}

We now will discuss the procedure to compute the transient of an MPL system via SMT-solving. The idea behind the SMT-based procedure is to transform the equality checking in line 11 of Algorithm \ref{comp_cyc_trans} into a formula in difference logic. Notice that the quantity $\textbf{M}[it]$ in Algorithm \ref{comp_cyc_trans} corresponds to $A^{\otimes it}\otimes V$ next, and $\mathsf{cone}(V)$ can be expressed as matrix $V$. Thus, it can be  written as
\begin{equation}
\label{eq_checking}
    (A^{\otimes it}\otimes V) \otimes \textbf{x}=(\lambda \times m)\otimes (A^{\otimes it-m}\otimes V)\otimes \textbf{x},~~\forall \textbf{x}\in \mathbb{R}^p,
\end{equation}
where $p$ is the number of columns of $V$. By denoting $R=A^{\otimes it}\otimes V$ and $S=(\lambda \times m)\otimes  A^{\otimes it-m}\otimes V$, \eqref{eq_checking} can be expressed as
\begin{equation}
\label{eq_func}
\bigwedge_{k=1}^n\!\left(\!\!\!\left(\bigoplus_{i=1}^p (\mathtt{x}_i+r_{ki})\geq \bigoplus_{j=1}^p (\mathtt{x}_j+s_{kj})\!\right)\!\!\wedge\!\!
\left(\bigoplus_{i=1}^p (\mathtt{x}_i+s_{ki})\geq \bigoplus_{j=1}^p (\mathtt{x}_j+r_{kj})\right)\!\!\!\right)\!,
\end{equation}
where $r_{ki}$ (resp. $s_{ki}$) is the element of $R$ (resp. $S$) at row $k$ and column $i$. For simplicity, we denote \eqref{eq_func} as $\mathtt{Eq}\mathtt{Func}(R,S)$. By Proposition \ref{reduced_formula}, each disjunct in \eqref{eq_func} can be expressed as a formula in difference logic. 

Algorithm \ref{trans_cone_smt} summarizes the SMT-based version of Algorithm \ref{comp_cyc_trans}. If the corresponding eigenspace of the matrix is not empty, we set the value for transient and cyclicity respectively to $k_0=0$ and $c=1$ (the smallest possible for both). Then, we generate the corresponding difference logic formula $F$ w.r.t. \eqref{eq_checking} in line 10. To check the validity of $F$, we use an SMT solver to check the unsatisfiability of the negation. If it is not satisfiable then the original formula is valid, and then we obtain the transient and cyclicity from the current value of $k_0$ and $c$.

On the other hand, if it is satisfiable then there exists a counterexample falsifying formula $F$. We express the counterexample from a satisfying assignment of $ \neg F$ as a real-valued vector $w\in \mathbb{R}^p$ (line 15). Vector $v=V\otimes w$ corresponds to the counterexample: its transient is greater than $k_0$ or its cyclicity is greater than $c$. The resulting transient and cyclicity of $v$ become the updated value for $(k_0,c)$. This process is repeated until either the SMT solver reports ``unsatisfiable'' in line 12 or $k_0+c$ exceeds the maximum bound $N$.  
(which corresponds to the termination condition of Algorithm \ref{comp_cyc_trans}).

\begin{algorithm}
\captionsetup{format=hang,width=\linewidth}
\caption{\footnotesize Computation of transient and cyclicity of $A$ w.r.t. $\mathsf{cone}(V)$ via SMT-solving}
\label{trans_cone_smt}
\footnotesize
\begin{algorithmic}[1]
\Function{TransConeSMT}{$A,V,N$}
\State $\chi\gets \Call{CycleTimeVector}{A}$
\If{elements of $\chi$ are all equal}
\State{$n \gets $ \Call{NrRows}{$A$}} 
\State{$p \gets $ \Call{NrCols}{$V$}} 
 \For{$i \in \{1 \cdots p\}$}
       \State{$x[i] \gets $ \Call{MakeSMTRealVar}{\ }} \Comment{symbolic variables}
  \EndFor
\State $k_0\gets 0,c\gets 1$
\While{$((k_0+c)\leq N)$}
\State $F \gets \mathtt{Eq}\mathtt{Func} (A^{\otimes k_0+c}\otimes V,(\lambda \times c)\otimes A^{\otimes k_0}\otimes V)$ 
\State $model \gets \Call{GetSMTModel}{ \neg F}$ 
\If{$model=\bot$} \Comment{formula is unsatisfiable}
\State \textbf{return} $\tuple{k_0,c}$
\Else \Comment{formula is satisfiable}
\State $w\gets\tuple{model(x[1]), \cdots model(x[p])}$\Comment{vector in $\mathbb{R}^p$}
\State $v\gets V\otimes w$ \Comment{vector in $\mathbb{R}^n$}
\State{$\tuple{k_0^\prime, c^\prime} \gets $ \Call{TransCone}{$A$, $A^{\otimes k_0}\otimes v$}} \Comment{computed by Algorithm \ref{comp_cyc_trans}}
\State $k_0\gets k_0+k_0^\prime$
\State $c\gets \texttt{LCM}(c,c^\prime)$
\EndIf
\EndWhile
\If{$((k_0+c)> N)$}
\State \textbf{print} ``terminated after reaching maximum bound''
\EndIf
\Else
\State \textbf{print} ``the transient does not exist''
\EndIf
\EndFunction
\end{algorithmic}
\end{algorithm}

Unlike Algorithm \ref{comp_cyc_trans}, which only works on max-plus cones, Algorithm \ref{trans_cone_smt} can be modified (into Algorithm \ref{comp_cyc_trans_smt}) so that it can be applied on any set of initial conditions $X\subseteq \mathbb{R}^n$. Although \eqref{eq_func} is can be translated exclusively to \rdl, we can extend $X$ as an \lra formula. In line 9 of Algorithm \ref{comp_cyc_trans_smt}, we generate a formula $F$ which corresponds to the equality checking between $A^{\otimes k_0}$ and $A^{\otimes k_0+c}$. If $X\rightarrow F$ is valid then we have $\mathsf{tr}(A,X)= k_0$ and $\mathsf{cyc}(A,X)= c$. Again, to check the validity of $X\rightarrow F$, we check the unsatisfiability of its negation.

\begin{algorithm}
\captionsetup{format=hang,width=\linewidth}
\caption{\footnotesize Computation of transient and cyclicity of $A$ w.r.t. a set of initial \hspace{6ex} conditions $X$ via SMT-solving}
\label{comp_cyc_trans_smt}
\footnotesize
\begin{algorithmic}[1]
\Function{TransSMT}{$A,X,N$}
\State $\chi\gets \Call{CycleTimeVector}{A}$
\If{elements of $\chi$ are all equal}
\State{$n \gets $ \Call{Row}{$A$}} \Comment{number of rows of $A$}
 \For{$i \in \{1 \cdots n\}$}
       \State{$x[i] \gets $ \Call{MakeSMTRealVar}{\ }} \Comment{symbolic variables}
  \EndFor
\State $k_0\gets 0,c\gets 1$
\While{$(k_0+c)\leq 1000$}
\State $F \gets \mathtt{Eq}\mathtt{Func} (A^{\otimes k_0+c},(\lambda\times c)\otimes A^{\otimes k_0})$ 
\State $model \gets \Call{GetSMTModel}{X\wedge \neg F}$ 
\If{$model=\bot$} \Comment{formula is unsatisfiable}
\State \textbf{return} $\tuple{k_0,c}$
\Else \Comment{formula is satisfiable}
\State $v\gets\tuple{model(x[1]), \cdots model(x[N])}$
\State{$\tuple{k_0^\prime, c^\prime} \gets $ \Call{TransCone}{$A$, $A^{\otimes k_0}\otimes v$}}
\State $k_0\gets k_0+k_0^\prime$
\State $c\gets \texttt{LCM}(c,c^\prime)$
\EndIf
\EndWhile
\If{$((k_0+c)> N)$}
\State \textbf{print} ``terminated after reaching maximum bound''
\EndIf
\Else
\State \textbf{print} ``the transient does not exist''
\EndIf
\EndFunction
\end{algorithmic}
\end{algorithm}
\subsection{A Synthesis Problem} 
In addition to computing the transient and cyclicity of $A\in \Rm{n}{n}$ w.r.t. a set of initial conditions, we show that by means of difference logic and SMT, one can synthesise sets of states corresponding to specific transient (and cylicity) defined as follows
\begin{eqnarray}
\label{synthesize1}
\mathcal{S}_{p,q}(A)&=&\{x\in \mathbb{R}^n\mid \mathsf{tr}(A,x)=p, \mathsf{cyc}(A,x)=q\},\\
\label{synthesize2}
\mathcal{S}_{p}(A)&=&\{x\in \mathbb{R}^n\mid \mathsf{tr}(A,x)=p\}.
\end{eqnarray}
On the one hand, the computation of \eqref{synthesize2} has been discussed in \cite[Section 4.2]{DiekyBackward} by applying  backward reachability analysis. On the other hand, to the best of the authors' knowledge, there is no approach to generate \eqref{synthesize1}. The following proposition shows that both \eqref{synthesize1} and \eqref{synthesize2} can be computed symbolically by expressing them as difference logic formulae: the set \eqref{synthesize1} (resp. \eqref{synthesize2}) is not empty if and only if the corresponding formula \eqref{Sp} (resp. \eqref{Spq}) is satisfiable.

\begin{proposition}
\normalfont
\label{prop_synthe}
Given $A\in \Rm{n}{n}$ with global cyclicity $c$ and maximum eigenvalue $\lambda$, we have 
\begin{equation}
\label{Sp}
    \mathcal{S}_p(A)=\left\{\!\!
    \begin{array}{ll}
       \mathsf{Eq}\mathsf{Func}(A^{\otimes p+c},\lambda c \otimes A^{\otimes p}),  &  \text{if}~p=0,\\\\
       \begin{array}{l}
       \mathsf{Eq}\mathsf{Func}(A^{\otimes p+c},\lambda c \otimes A^{\otimes p})\wedge\\  ~~~\neg\mathsf{EqFunc}(A^{\otimes p-1+c},\lambda c \otimes A^{\otimes p-1}), 
       \end{array}
       & \text{if}~p>0,
    \end{array}\right.
\end{equation}
and
\begin{equation}
\label{Spq}
    \mathcal{S}_{p,q}(A)= \left\{\!\!
    \begin{array}{ll}
    \begin{array}l
         \mathsf{EqFunc}(A^{\otimes p+q},\lambda q \otimes A^{\otimes p})\wedge\\ 
        ~~~\displaystyle  \bigwedge_{d\in \texttt{Div}(q)-\{q\}} \neg \mathsf{EqFunc}(A^{\otimes p+d},\lambda d \otimes A^{\otimes p}), 
    \end{array}
    &  \text{if}~p=0,\\\\
    \begin{array}{l}
\displaystyle \mathsf{EqFunc}(A^{\otimes p+q},\lambda q \otimes A^{\otimes p})\wedge \neg \mathsf{EqFunc}(A^{\otimes p-1+q},\lambda q \otimes A^{\otimes p-1})\wedge\\
\displaystyle \bigwedge_{d\in \texttt{Div}(q)-\{q\}}\hspace*{-4ex}\neg \mathsf{EqFunc}(A^{\otimes p+d},\lambda d \otimes A^{\otimes p}),  
    \end{array}
&  \text{if}~p>0,\\
    \end{array}\right.
\end{equation}
where $\texttt{Div}(q)$ is a set of divisors of $q$. 
\end{proposition}
\begin{proof}
Notice that, the formula $\mathsf{Eq}\mathsf{Func}(A^{\otimes p+c},\lambda c \otimes A^{\otimes p})$ corresponds to a periodic behavior at bound $p$ and $p+c$ with cyclicity at most $c$. More precisely, if $x(0)$ satisfies $\mathsf{Eq}\mathsf{Func}(A^{\otimes p+c},\lambda c \otimes A^{\otimes p})$ then $x(p+c)=\lambda c \otimes x(p)$ or in other words $\mathsf{tr}(A,x(0))\leq p$. On the other hand,
if $x(0)$ satisfies $\neg \mathsf{Eq}\mathsf{Func}(A^{\otimes p-1+c},\lambda c \otimes A^{\otimes p-1})$ then $\mathsf{tr}(A,x(0))>p-1$. 

Similarly, $ \mathsf{EqFunc}(A^{\otimes p+q},\lambda q \otimes A^{\otimes p})\wedge \neg \mathsf{EqFunc}(A^{\otimes p-1+q},\lambda q \otimes A^{\otimes p-1})$ corresponds to a periodic behavior with transient $p$ and cyclicity at most $q$. The remaining conjuncts guarantee that the cyclicity cannot be smaller than $q$.\QEDB
\end{proof}

As both \eqref{synthesize1} and \eqref{synthesize2} can be expressed as formulae in difference logic, the problem of determining the emptiness of both sets is decidable. 
By definition, for \textit{never periodic} MPL system, $\mathcal{S}_{p,q}=\mathcal{S}_p=\emptyset$ for all $p,q$. 
Furthermore, for irreducible MPL systems the emptiness of \eqref{synthesize1} and \eqref{synthesize2} is related to the global transient and cyclicity of $A$.
\begin{proposition}
\normalfont
For an irreducible matrix $A\in \Rm{n}{n}$ with global transient $k_0$ and cyclicity $c$ we have $\mathcal{S}_0(A)=E(A^{\otimes c})$ and $\mathcal{S}_{0,1}(A)=E(A)$. Furthermore,
\begin{itemize}
    \item[i.] $\mathcal{S}_p(A)\neq \emptyset$ iff $p\leq k_0$, 
    \item[ii.] If $p> k_0$ or $q$ is not a divisor of $c$ then $\mathcal{S}_{p,q}(A)= \emptyset$,
    \item[iii.] If $\mathcal{S}_{p,q}(A)$ is empty then so is $\mathcal{S}_{p+1,q}(A)$.  
\end{itemize}
\end{proposition}
\begin{proof}
Let us assume the eigenvalue for $A$ is $\lambda$. Notice that, $\textbf{x}(0)\in \mathcal{S}_0(A)$ if and only if $\textbf{x}(c)=(\lambda\times c)\otimes \textbf{x}(0)$, or equivalently $A^{\otimes c}\otimes \textbf{x}(0)=(\lambda\times c)\otimes \textbf{x}(0)$. This shows that $\textbf{x}(0)$ is an eigenvector of $A^{\otimes c}$. Thus, $ \mathsf{S}_0(A)= E(A^{\otimes c})$. The proof for $\mathcal{S}_{0,1}(A)=E(A)$ can obtained similarly.
\begin{itemize}
    \item[i.] As the global transient for $A$ is $k_0$, it is by default that $\mathcal{S}_p(A)=\emptyset$ for $p>k_0$. Suppose $\textbf{x}(0)\in \mathbb{R}^n$ such that $\mathsf{tr}(A,\textbf{x}(0))=k_0$. It is straightforward that $\textbf{x}(p)\in\mathcal{S}_{p-k_0}(A)$ for $p\leq k_0$. This completes the proof.
    \item[ii.] If $p>k_0$ we have $\mathcal{S}_p(A)=\emptyset$ which implies $\mathcal{S}_{p,q}=\emptyset$. Suppose $q$ is not a divisor of $c$ and $\mathcal{S}_{p,q}\neq\emptyset$. Then there exists $\textbf{x}(0)$ such that $\mathsf{cyc}(A,\textbf{x}(0))=\texttt{lcm}(c,q)>c$. This contradicts the fact that $c$ is the global cyclicity. 
    \item[iii.] The proof is from the fact that if $\textbf{x}\in \mathcal{S}_{p+1,q}(A)$ then $A\otimes \textbf{x}\in \mathcal{S}_{p,q}(A)$. \QEDB
    \end{itemize}
\end{proof}
\begin{example}
\label{rem_synthe}
Let us recall the $3\times 3$ MPL system in Example \ref{ex_trans} with $a=8$.
From the precedence graph $\mathcal{G}(B)$, the global cyclicity is $c=2$. Leaving details aside, for $p\geq 1$, we have 
\[\mathsf{EqFunc}(B^{p+2},18 \otimes B^{\otimes p})\equiv\left\{
\begin{array}{ll}
     (\mathtt{x}_1-\mathtt{x}_3\geq 6-p)\vee (\mathtt{x}_2-\mathtt{x}_3\geq 8-p), &\text{if}~p~\text{is odd},\\
     (\mathtt{x}_1-\mathtt{x}_3\geq 7-p) \vee (\mathtt{x}_2-\mathtt{x}_3\geq 7-p), &\text{if}~p~\text{is even}.
\end{array}
\right.\]
Thus, by \autoref{prop_synthe}, for $p\geq 2$ we have
\begin{equation}
\nonumber
    \mathcal{S}_p(B)\!=\!\left\{\!
    \begin{array}{ll}
        \!\{\textbf{x}\in \mathbb{R}^3\!\!\mid\!\! (6-p\leq x_1-x_3<8-p) \wedge (x_2-x_3<8-p)\}, &\text{if}~p~\text{is odd},\\
        \!\{\textbf{x}\in \mathbb{R}^3\!\!\mid\!\! (x_1-x_3<7-p)\wedge(7-p\leq x_2-x_3<9-p)\}, &\text{if}~p~\text{is even}.\\
    \end{array}
     \right.
\end{equation}
An illustration of the above sets is depicted in \autoref{fig1}. 
From an initial vector $\textbf{x}(0)=[4~2~0]^\top \in \mathcal{S}_3(B)$ one can compute $\textbf{x}(k)$ for $k=1,\ldots,5$ as follows: 
\[
\begin{bmatrix}
    10\\14\\8
\end{bmatrix},
\begin{bmatrix}
    22\\20\\16
\end{bmatrix},
\begin{bmatrix}
    28\\32\\25
\end{bmatrix},
\begin{bmatrix}
    40\\38\\33
\end{bmatrix},
\begin{bmatrix}
    46\\50\\43
\end{bmatrix}.
\]
Notice that, $\textbf{x}(5)=18\otimes \textbf{x}(3)$ which confirms that $\mathsf{tr}(B,\textbf{x}(0))=3$. It is straightforward to conclude that the global transient for $B$ does not exist. \QEDB
\end{example}
\begin{figure}[!ht]
\centering
\begin{tikzpicture}[node distance=0.7cm and 0.7cm,scale=0.475]
\footnotesize
\fill[black!10!white] (9,9) rectangle (-10,-5);
\fill[black!30!white] (5,7) rectangle (-10,-5);
\draw[-,dashed](4.95,6.95)--(4.95,-5);
\draw[-,dashed] (2.95,4.95)--(2.95,-5);
\draw[-,dashed] (0.95,2.95)--(0.95,-5);
\draw[-,dashed] (-1.05,0.95)--(-1.05,-5);
\draw[-,dashed] (-3.05,-1.05)--(-3.05,-5);
\draw[-,dashed] (-5.05,-3.05)--(-5.05,-5);

\draw[-,dashed] (5-0.05,6.95)--(-10,6.95);
\draw[-,dashed] (5-0.05,4.95)--(-10,4.95);
\draw[-,dashed] (3-0.05,2.95)--(-10,2.95);
\draw[-,dashed] (1-0.05,0.95)--(-10,0.95);
\draw[-,dashed] (-1-0.05,-0.95)--(-10,-0.95);
\draw[-,dashed] (-3-0.05,-2.95)--(-10,-2.95);
\draw[-,dashed] (-5-0.05,-4.95)--(-10,-4.95);

\draw[-] (5.05,7.05)--(5.05,-5);
\draw[-] (3.05,5.05)--(3.05,-5);
\draw[-] (1.05,3.05)--(1.05,-5);
\draw[-] (-0.95,1.05)--(-0.95,-5);
\draw[-] (-2.95,-0.95)--(-2.95,-5);
\draw[-] (-4.95,-2.95)--(-4.95,-5);
\draw[-] (-6.95,-4.95)--(-6.95,-5);

\draw[-] (5.05,7.05)--(-10,7.05);
\draw[-] (5.05,5.05)--(-10,5.05);
\draw[-] (3.05,3.05)--(-10,3.05);
\draw[-] (1.05,1.05)--(-10,1.05);
\draw[-] (-0.95,-0.95)--(-10,-0.95);
\draw[-] (-2.95,-2.95)--(-10,-2.95);
\draw[-] (-4.95,-4.95)--(-10,-4.95);
\draw[->,>=stealth] (-10,0) -- (9,0);
\draw[->,>=stealth] (0,-5) -- (0,9);
\foreach \x in {-4,-2,2,4,...,8}{
	\draw[-] (0.1,\x)--(-0.1,\x);	
	\node at (-0.3,\x) {\scriptsize \x};
}
\foreach \x in {-8,-6,...,-2,2,4,...,8}{
	\draw[-] (\x,0.1)--(\x,-0.1);
	\node at (\x,-0.4) {\scriptsize \x};
}
	\node at (9.55,0) {$x_1$};
	\node at (0,9.4) {$x_2$};
    \node at (7,3) { $\mathcal{S}_0(B)~\cup$};
    \node at (7,2.25) { $\mathcal{S}_1(B)$};
    \node at (-2,6) { $\mathcal{S}_2(B)$};
    \node at (-3,4) { $\mathcal{S}_4(B)$};
    \node at (-4,2) { $\mathcal{S}_6(B)$};
    \node at (-5,0.5) { $\mathcal{S}_8(B)$};
    \node at (-6,-2) { $\mathcal{S}_{10}(B)$};
    \node at (-7.5,-4) { $\mathcal{S}_{12}(B)$};
    \node at (4,1) {$\mathcal{S}_3(B)$};
    \node at (2,-2) {$\mathcal{S}_5(B)$};
    \node at (0,-3) {$\mathcal{S}_7(B)$};
    \node at (-2,-3) {$\mathcal{S}_9(B)$};
    \node at (-4,-4) {$\mathcal{S}_{11}(B)$};
\end{tikzpicture} 
\captionsetup{format=hang,width=0.85\linewidth}
\caption{Plots of the synthesized sets projected on the plane $x_3=0$. The solid and dashed lines represent $\geq$ and $>$, respectively.}
\label{fig1}
\end{figure}
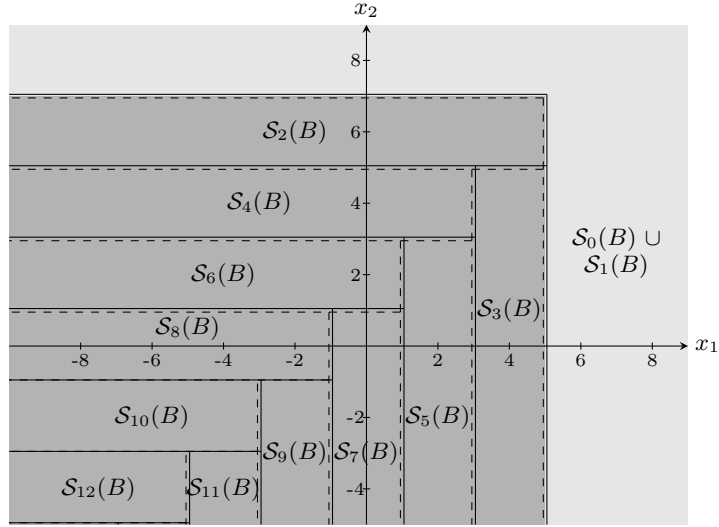
\section{Computational Benchmarks}
We compare the performance of Algorithms \ref{comp_cyc_trans} and \ref{comp_cyc_trans_smt}, to compute the transient of MPL systems. The experiments for both procedures are implemented in Python. For the SMT solver, we use Yices 2.2 \cite{Yices}. The computational benchmark has been implemented on an Intel\textregistered{} Xeon\textregistered{} CPU E5-1660 v3, 16 cores, 3.0GHz each, and 16GB of RAM. For the experiments, we generate 1000 irreducible matrices of dimension $n$, with $m$ finite elements in each row, where the values of the finite elements are rational numbers $\frac{p}{q}$ with $1\leq p\leq 100$ and $1\leq q\leq 5$. The locations of the finite elements are chosen randomly. We focus on irreducible matrices to ensure the termination of the algorithms. Algorithm \ref{comp_cyc_trans} is initialised by setting $V$ to be a max-plus identity matrix, while for Algorithm \ref{comp_cyc_trans_smt} the set of initial conditions is expressed as $X\equiv \mathtt{true}$. For all experiments, we choose $N=10000$ as the maximum bound. The benchmarks are stored at \url{https://es-static.fbk.eu/people/amicheli/resources/formats20/}, where we have chosen  $n\in\{4,6,8,10,20,30,40\}$ and three different values of $m$ for each $n$.

\begin{figure}[!ht]
  \begin{subfigure}[t]{.5\textwidth}
  \centering
  \includegraphics[scale=0.55]{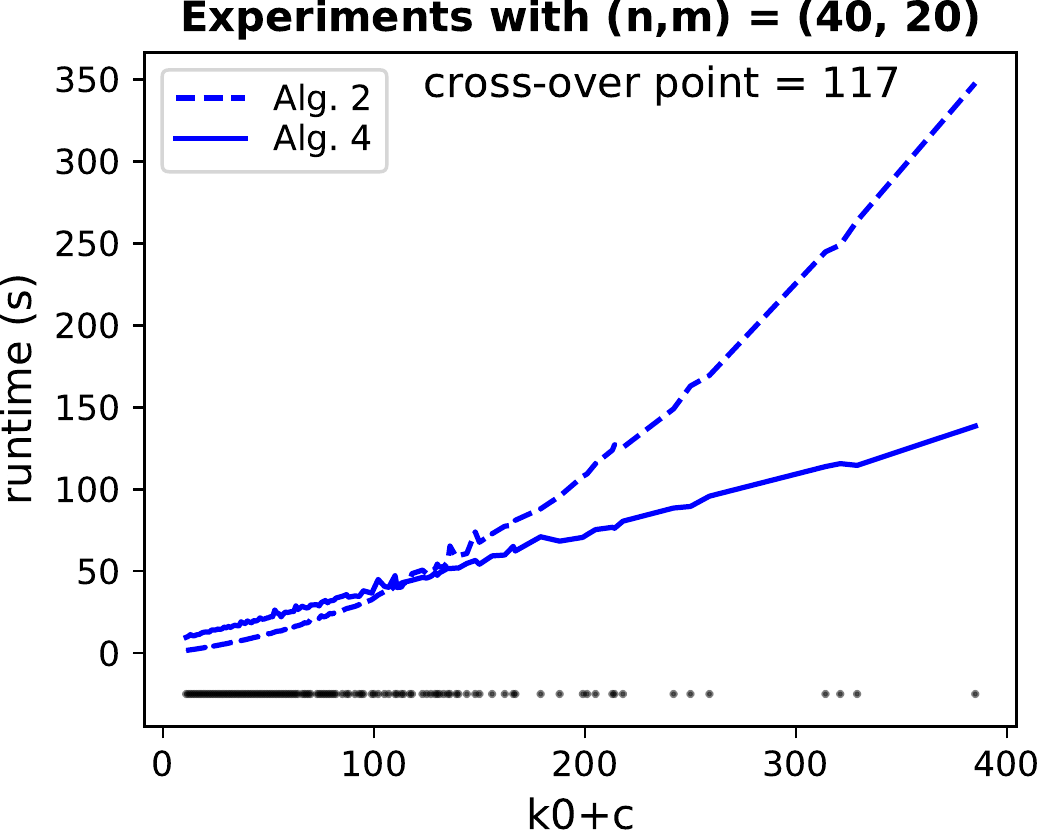}
  \caption{}
\end{subfigure}
\begin{subfigure}[t]{.5\textwidth}
  \centering
  \includegraphics[scale=0.55]{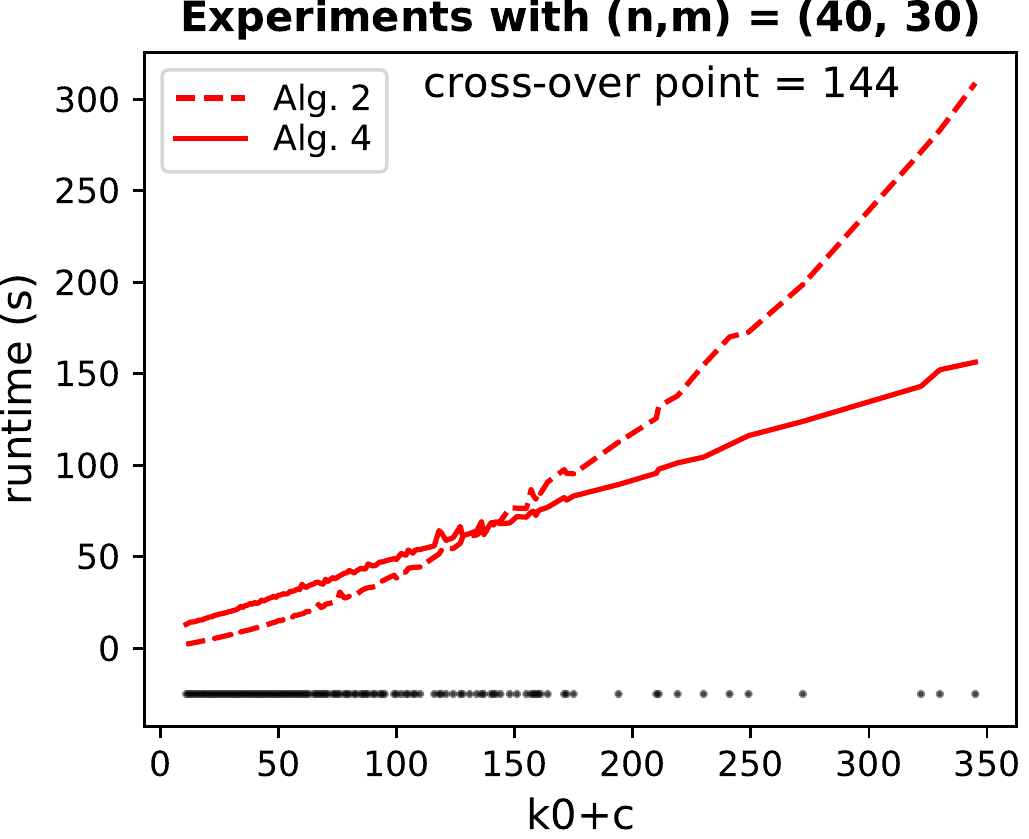}
  \caption{}
\end{subfigure}
\begin{subfigure}[t]{.5\textwidth}
  \centering
  \includegraphics[scale=0.55]{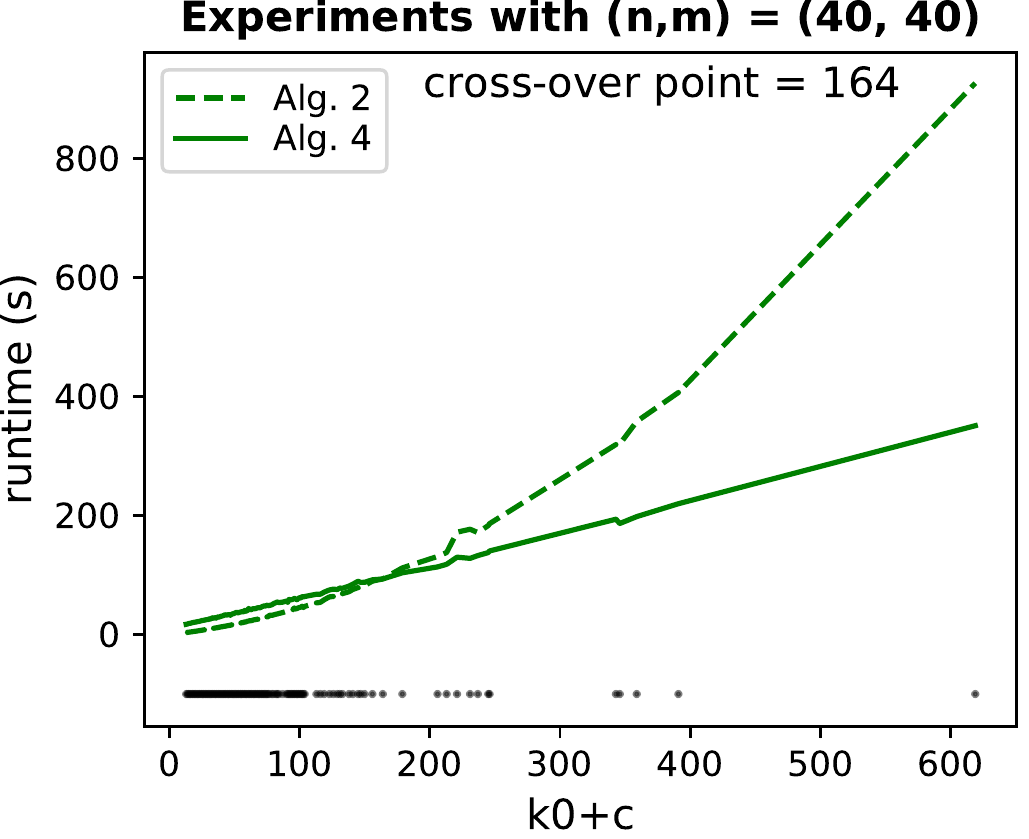}
  \caption{}
\end{subfigure}
\begin{subfigure}[t]{.5\textwidth}
  \centering
  \includegraphics[scale=0.55]{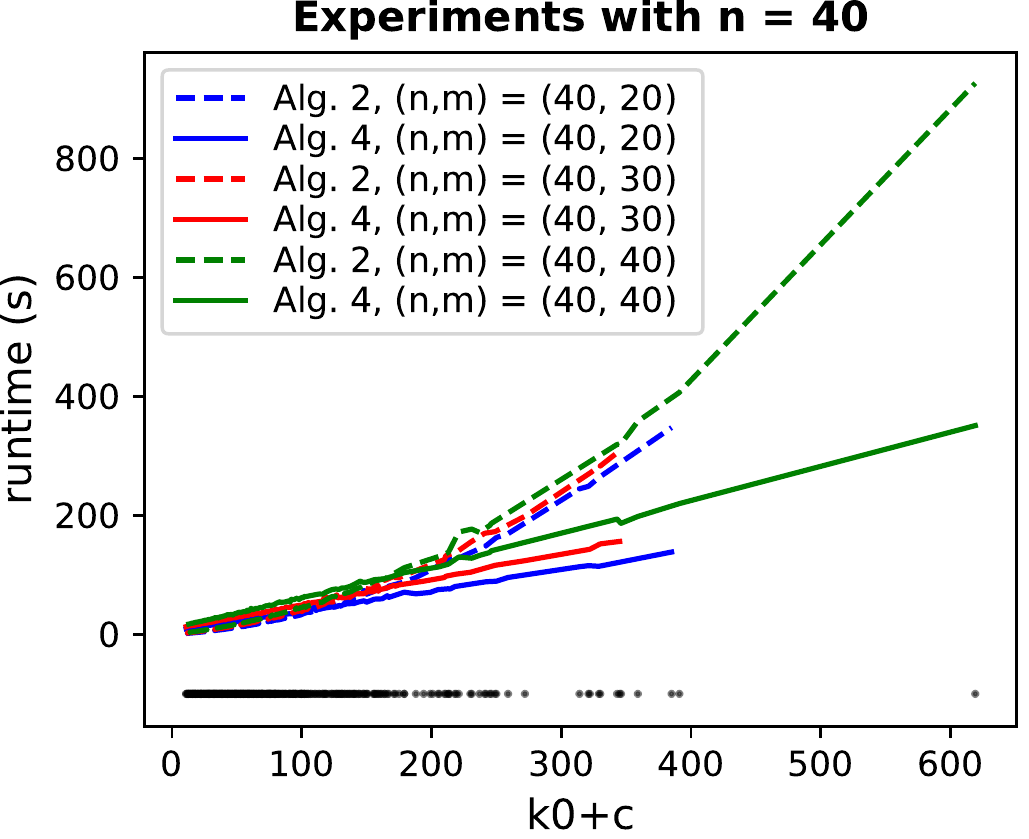}
  \caption{}
\end{subfigure}
    \captionsetup{format=hang,width=\linewidth}
    \caption{The plots of running time of Algorithms \ref{comp_cyc_trans} and  \ref{comp_cyc_trans_smt} from 1000 experiments with $n=40$ and $m\in \{20,30,40\}$. A ``cross-over point'' is the smallest value of $k_0+c$ when  Algorithm \ref{comp_cyc_trans_smt} is faster.} 
    \label{benchmark}
\end{figure}

\autoref{benchmark}(a)-(c) illustrate the experiments for $n=40$ and $m\in\{20,30,40\}$ (the experiments for other pairs $(n,m)$ are presented in the Appendix). They show the plots of the running times of Algorithm \ref{comp_cyc_trans} (dashed lines) and of Algorithm \ref{comp_cyc_trans_smt} (solid lines) against the resulting transient $k_0$ and cyclicity $c$ - the scattered plots (in black) correspond to the resulting $k_0+c$. 
If there are several experiments with the same value of $k_0+c$ then we display the average running time among those experiments. 
It is evident that most of the experiments result in small $k_0+c$.

With regards to the running time, the matrix-multiplication algorithm is faster when the values of $k_0+c$ are quite small. On the other hand, the larger the value of $k_0+c$, the better the performance of the SMT-based algorithm is. We argue that this is because in Algorithm \ref{comp_cyc_trans_smt} there may be a large increment from the current guess of transient and cyclicity to the new ones. Whereas in Algorithm \ref{comp_cyc_trans}, the next candidate of transient and cyclicity is increased by one at each iteration.

As depicted in \autoref{benchmark}(d), the number of finite elements $m$ clearly affects the running time of the algorithms. We recall that the running time of Algorithm \ref{comp_cyc_trans_smt} depends on the satisfaction checking of a difference logic formula in line 11. The more are the finite elements, the more likely the formula is complex, and therefore the slower is the associated running time. Interestingly, based on the outcomes of the benchmarks which are presented in the Appendix, the finite elements also affect the cross-over points, which tend to increase gradually as the number of finite elements grows larger.

\section{Conclusions and Future Work}



In this paper, we have introduced a novel, SMT-based approach to compute the transient of MPL systems: our technique encodes the problem as a sequence of satisfiability queries over formulae in difference logic, which can be solved by standard SMT solvers. We have also presented a procedure to partition the state-space of MPL systems w.r.t. a given transient and cyclicity pair.  
The procedure has been thoroughly tested on computational benchmarks and the results show how the SMT-based algorithm is much faster that state-of-the-art techniques to compute large values of transient and cyclicity. Furthermore, we highlight that the SMT-based method can be applied to compute the transient for any initial condition, as long as it is expressible as an \lra formula.

For future research, we are interested in exploring and developing SMT-based procedures for the general model checking of MPL systems.

%
%
%
\bibliographystyle{splncs04}
\bibliography{References}
\newpage

\begin{subappendices}
\renewcommand{\thesection}{\Alph{section}}%
\section{Appendix}
The following table summarizes the outcomes of the experiments for each pair $(n,m)$ w.r.t. the cross-over point and the maximum value of $k_0+c$ (denoted by $N^\ast$).
\begin{table}[!htt]
    \centering
    \caption{The summary of the experiments.}
    \begin{tabular}{|c|c|c|c|c|c|c|c|c|c|c|}
    \cline{1-3}\cline{5-7}\cline{9-11}
    $(n,m) $&cross-over& $N^\ast$ && $(n,m) $&cross-over& $N^\ast$&& $(n,m) $&cross-over& $N^\ast$\\
    \cline{1-3}\cline{5-7}\cline{9-11}
    (4,2) & 28 &1013&&(10,6)&52&659&&(30,10)&99&455\\
    \cline{1-3}\cline{5-7}\cline{9-11}
    (4,3) & 37 &349&&(10,8)&45&285&&(30,20)&103&698\\
    \cline{1-3}\cline{5-7}\cline{9-11}
    (4,4) & 42 &126&&(10,10)&47&733&&(30,20)&127&328\\
    \cline{1-3}\cline{5-7}\cline{9-11}
    (6,2) & 31 &5488&&(20,12)&69&273&&(40,20)&117&385\\
    \cline{1-3}\cline{5-7}\cline{9-11}
    (6,4) & 34 &640&&(20,16)&76&464&&(40,30)&144&345\\
    \cline{1-3}\cline{5-7}\cline{9-11}
    (6,6) & 42 &186&&(20,20)&87&478&&(40,40)&164&619\\
    \cline{1-3}\cline{5-7}\cline{9-11}
    (8,4) & 42 &4208\\
    \cline{1-3}
    (8,6) & 45 &220\\
    \cline{1-3}
    (8,8) & 46 &320\\
    \cline{1-3}
    \end{tabular}
\end{table}\\
The following figures illustrate the experiments for each pair $(n,m)$.
\begin{figure}[!ht]
  \begin{subfigure}[t]{.5\textwidth}
  \centering
  \includegraphics[scale=0.55]{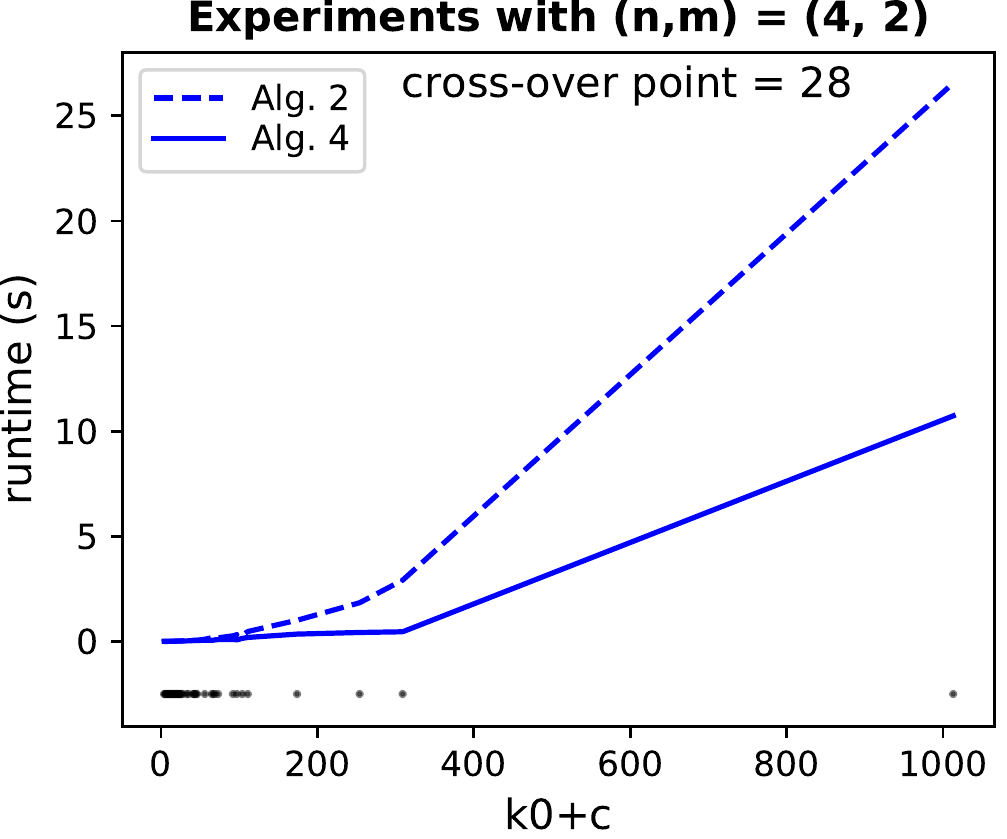}
\end{subfigure}
\begin{subfigure}[t]{.5\textwidth}
  \centering
  \includegraphics[scale=0.55]{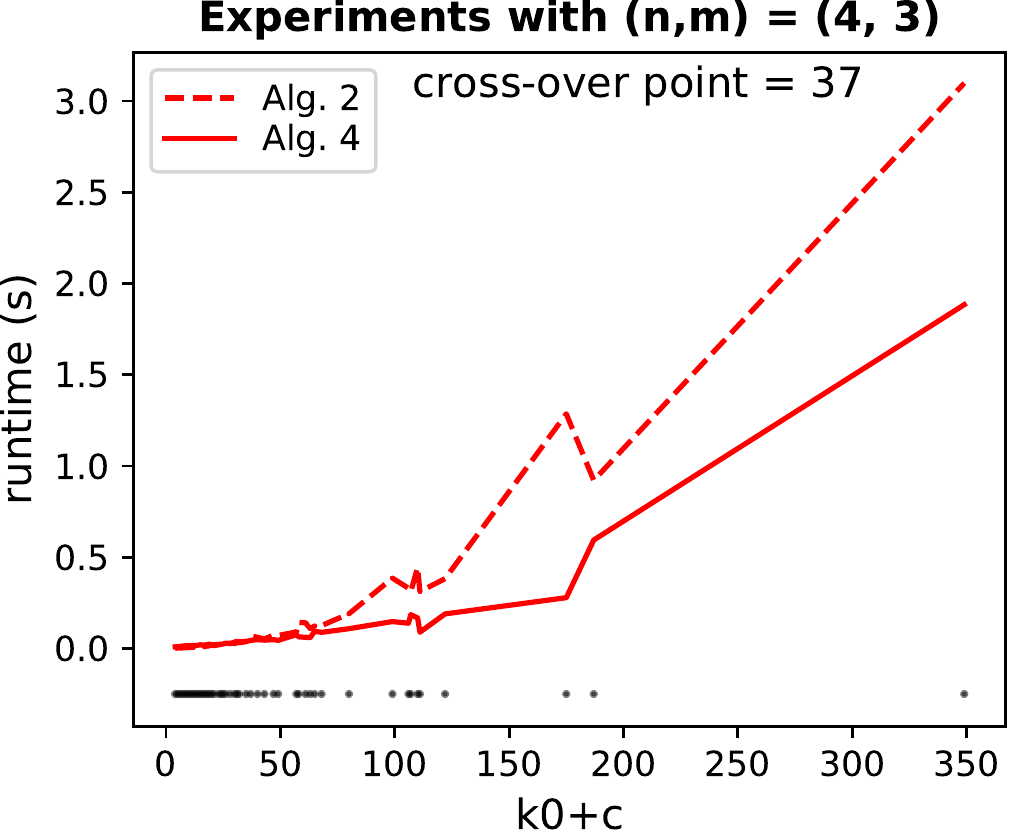}
\end{subfigure}
\vspace{0ex}\\
\begin{subfigure}[t]{.5\textwidth}
  \centering
  \includegraphics[scale=0.55]{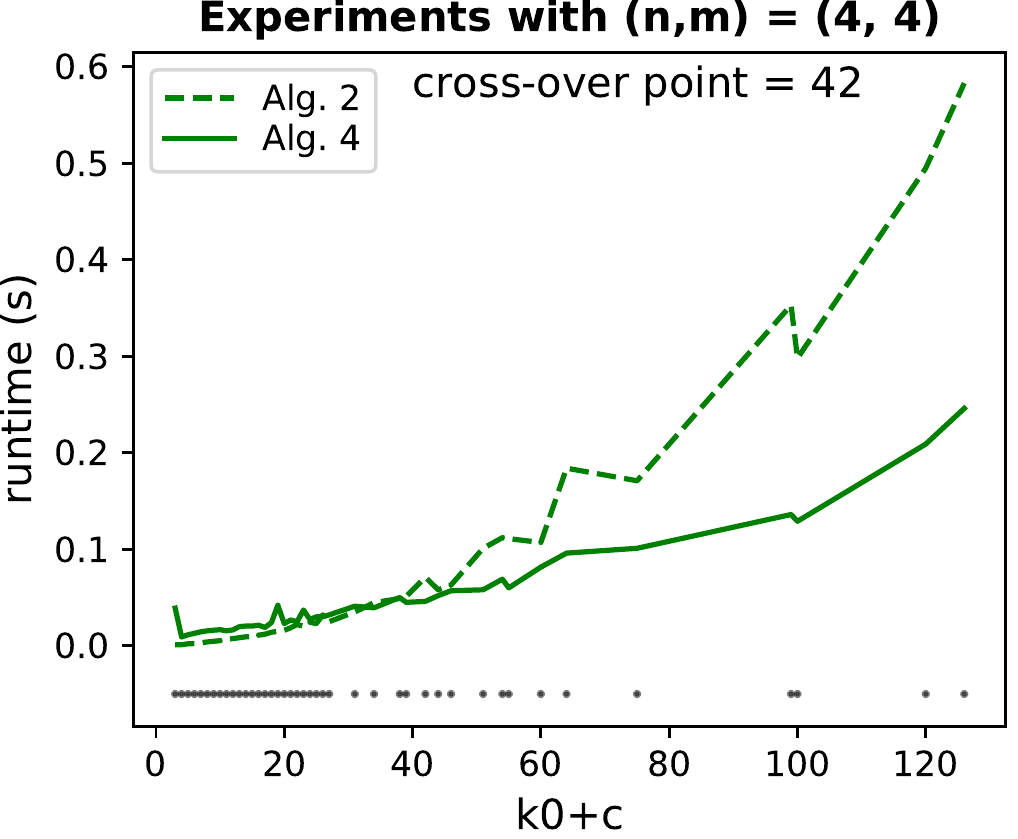}
\end{subfigure}
\begin{subfigure}[t]{.5\textwidth}
  \centering
  \includegraphics[scale=0.55]{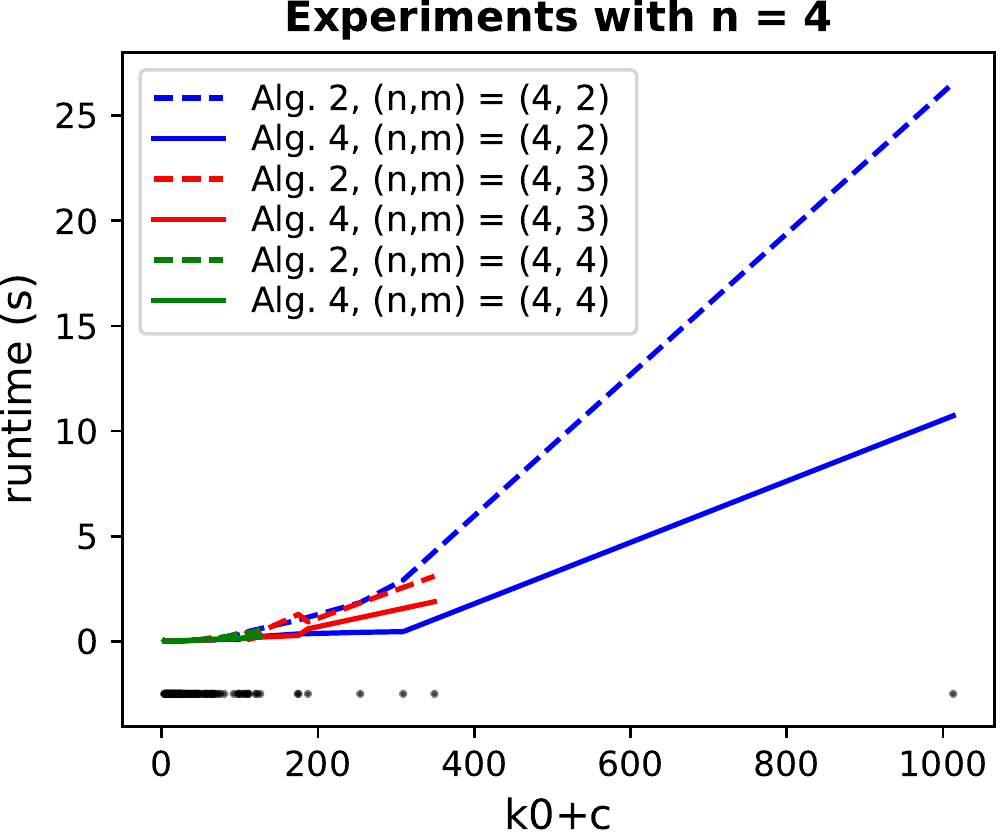}
\end{subfigure}
\end{figure}
\begin{figure}[!ht]
\begin{subfigure}[t]{.5\textwidth}
  \centering
  \includegraphics[scale=0.55]{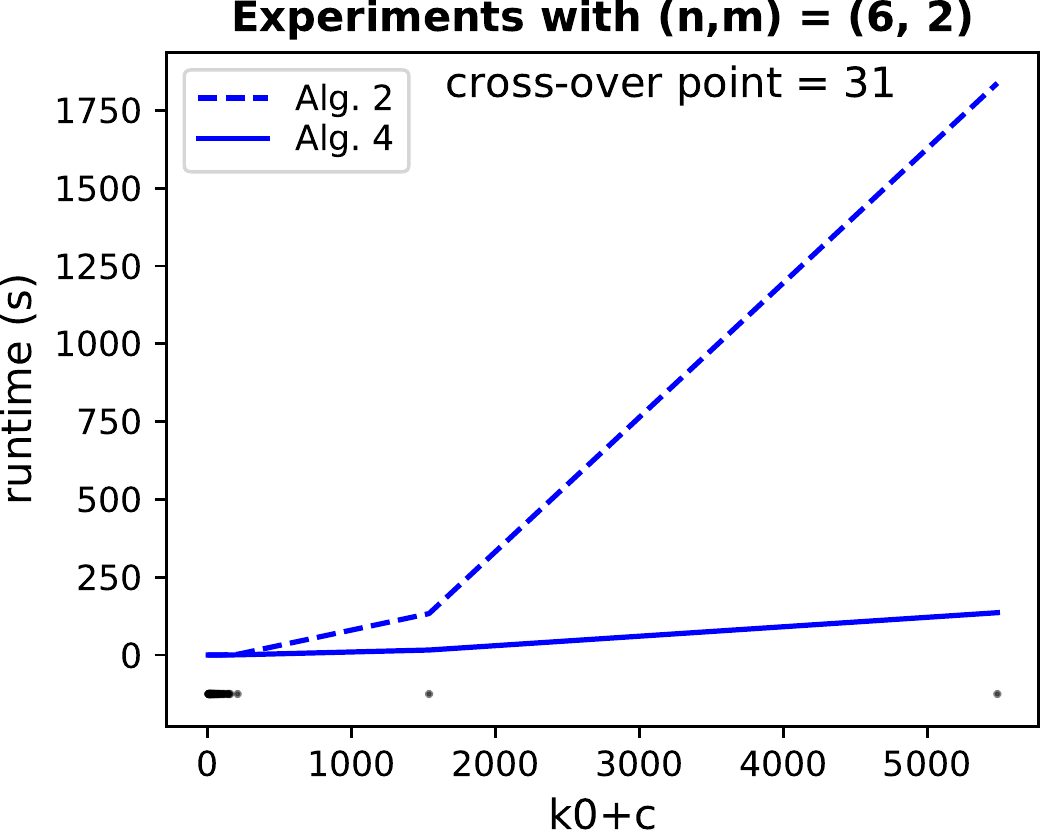}
\end{subfigure}
\begin{subfigure}[t]{.5\textwidth}
  \centering
  \includegraphics[scale=0.55]{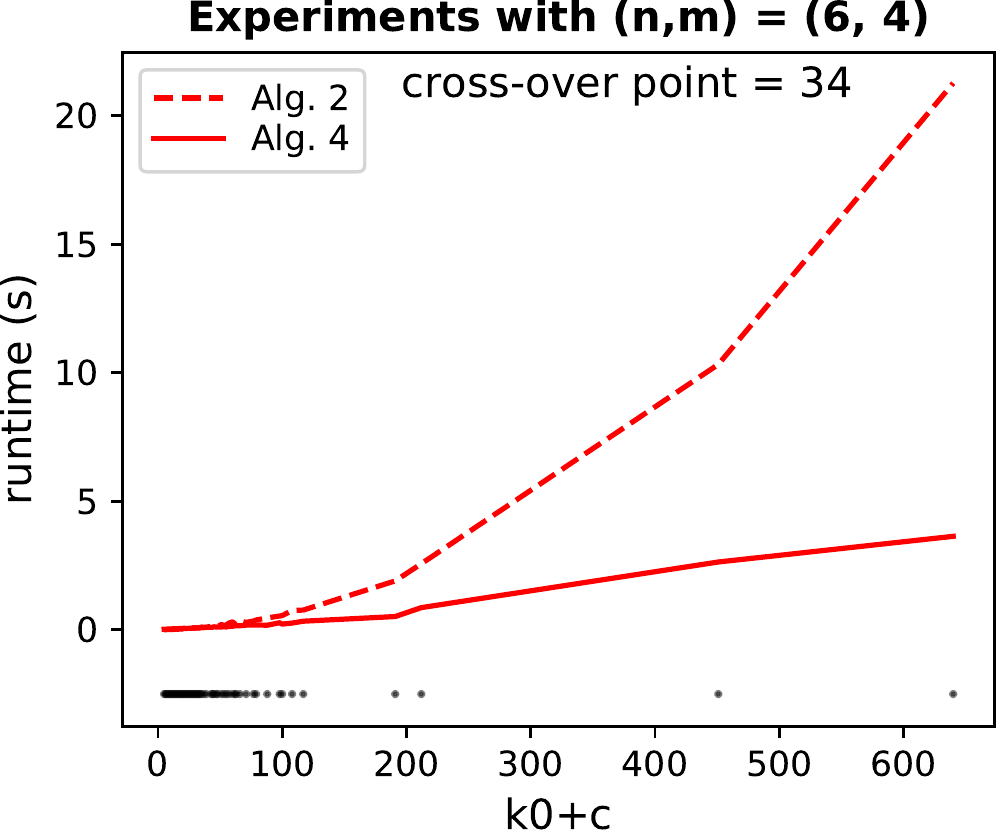}
\end{subfigure}
\vspace{0ex}\\
\begin{subfigure}[t]{.5\textwidth}
  \centering
  \includegraphics[scale=0.55]{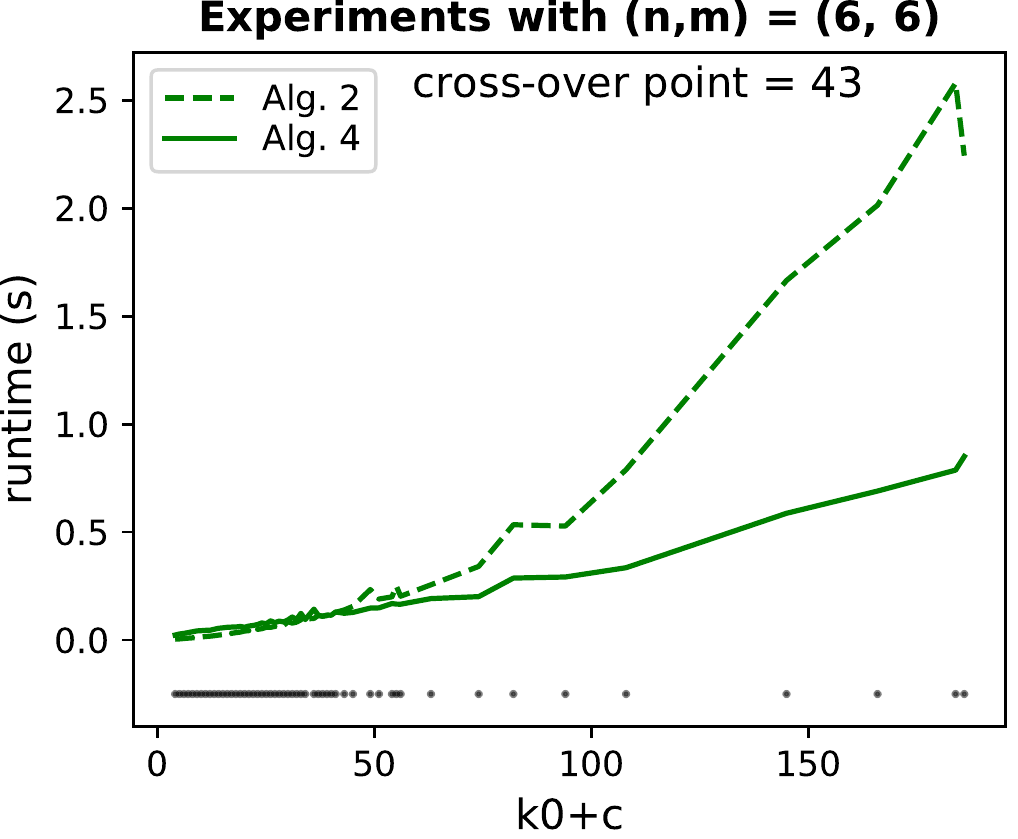}
\end{subfigure}
\begin{subfigure}[t]{.5\textwidth}
  \centering
  \includegraphics[scale=0.55]{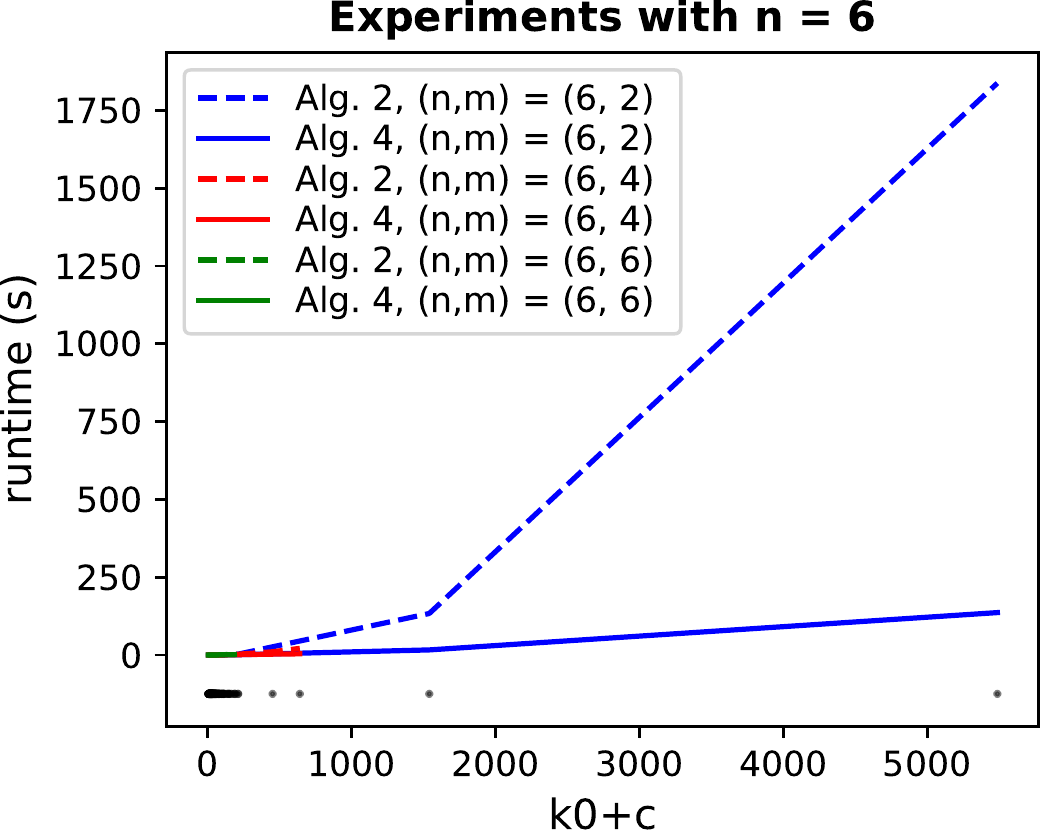}
\end{subfigure}
\vspace{0ex}\\
\begin{subfigure}[t]{.5\textwidth}
  \centering
  \includegraphics[scale=0.55]{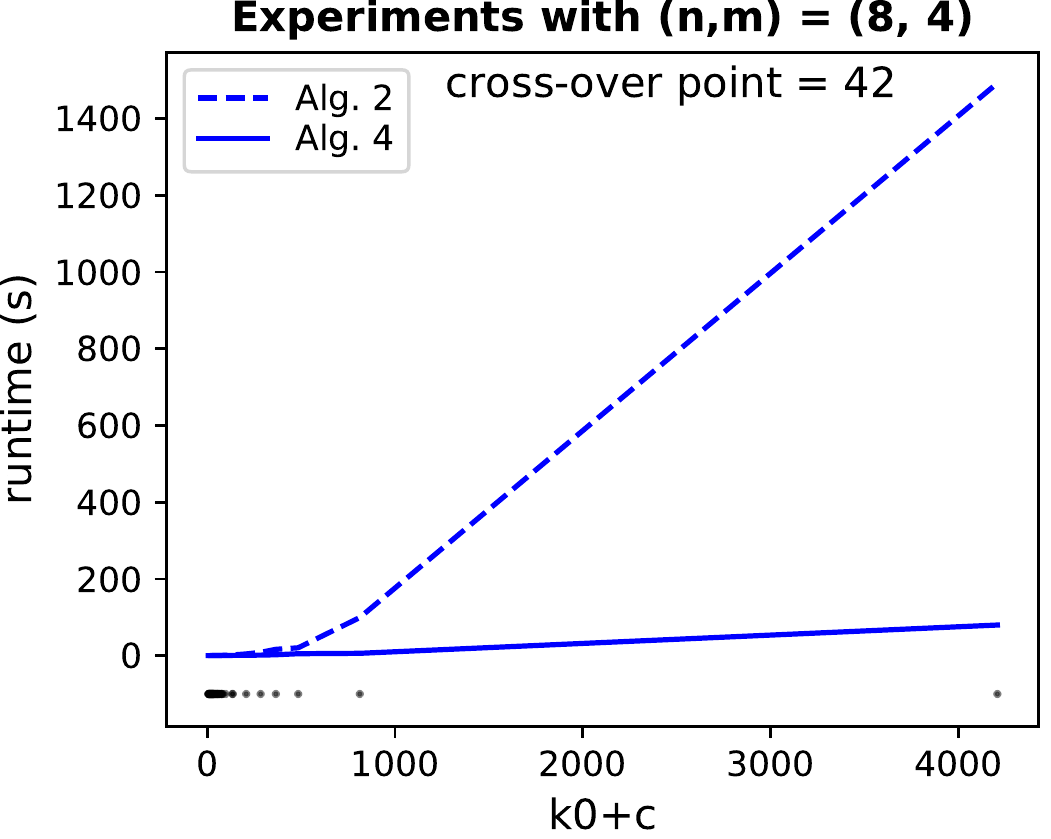}
\end{subfigure}
\begin{subfigure}[t]{.5\textwidth}
  \centering
  \includegraphics[scale=0.55]{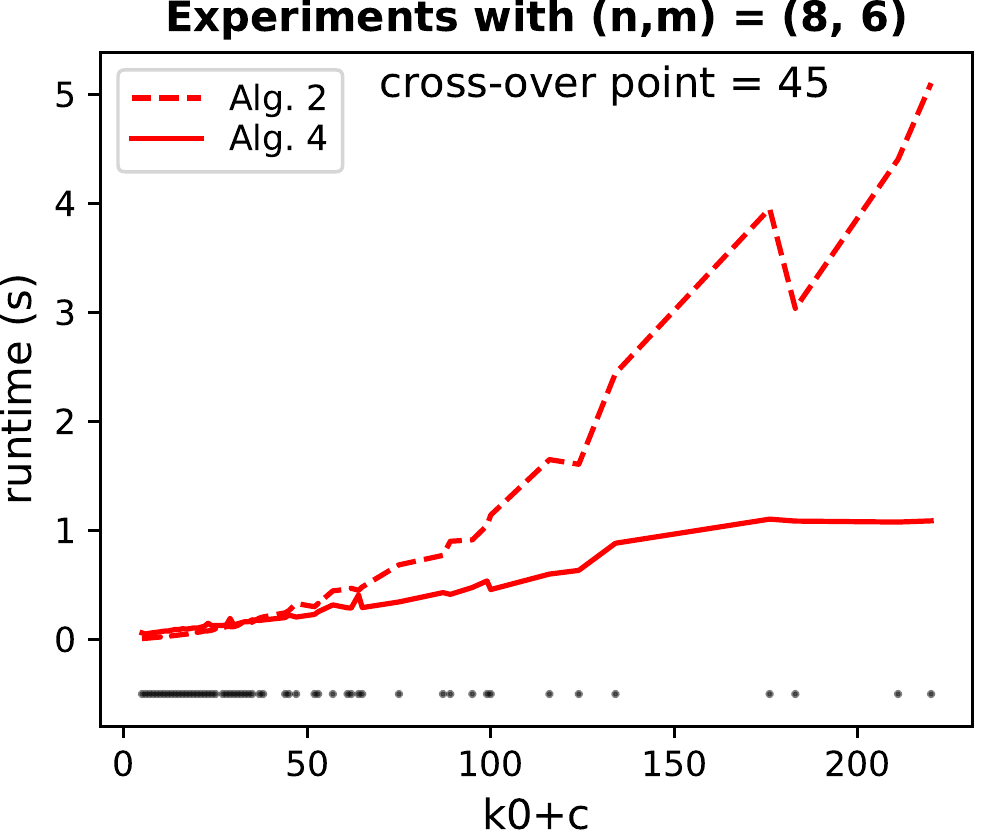}
\end{subfigure}
\vspace{0ex}\\
\begin{subfigure}[t]{.5\textwidth}
  \centering
  \includegraphics[scale=0.55]{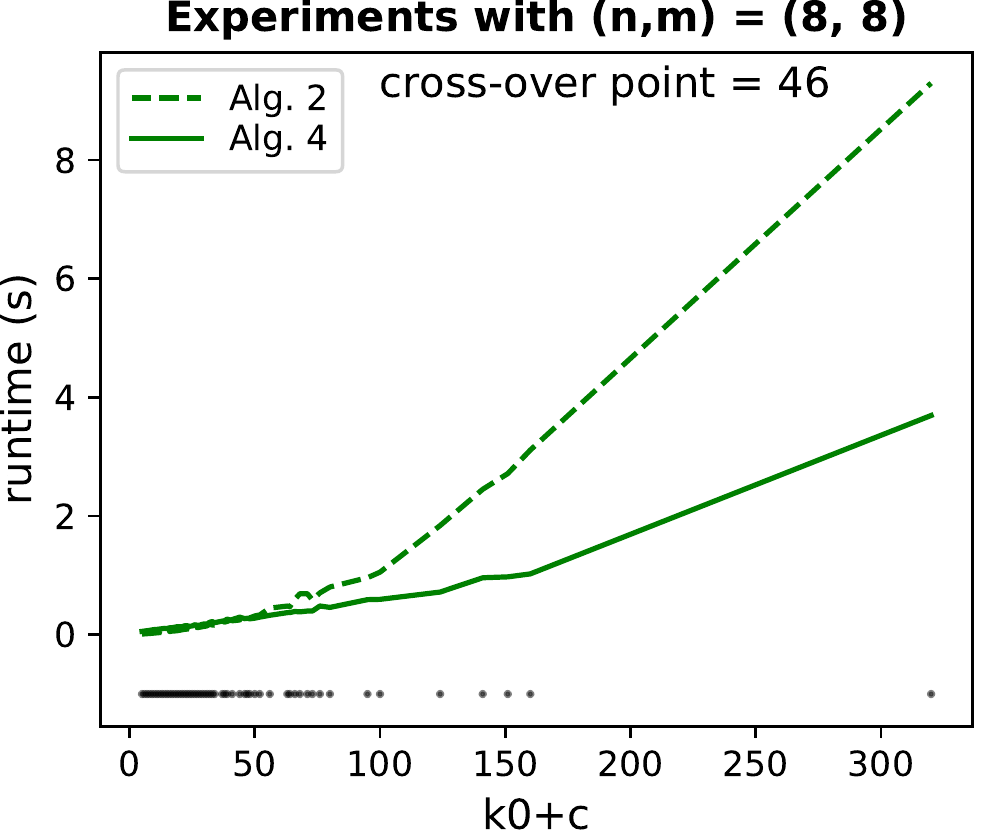}
\end{subfigure}
\begin{subfigure}[t]{.5\textwidth}
  \centering
  \includegraphics[scale=0.55]{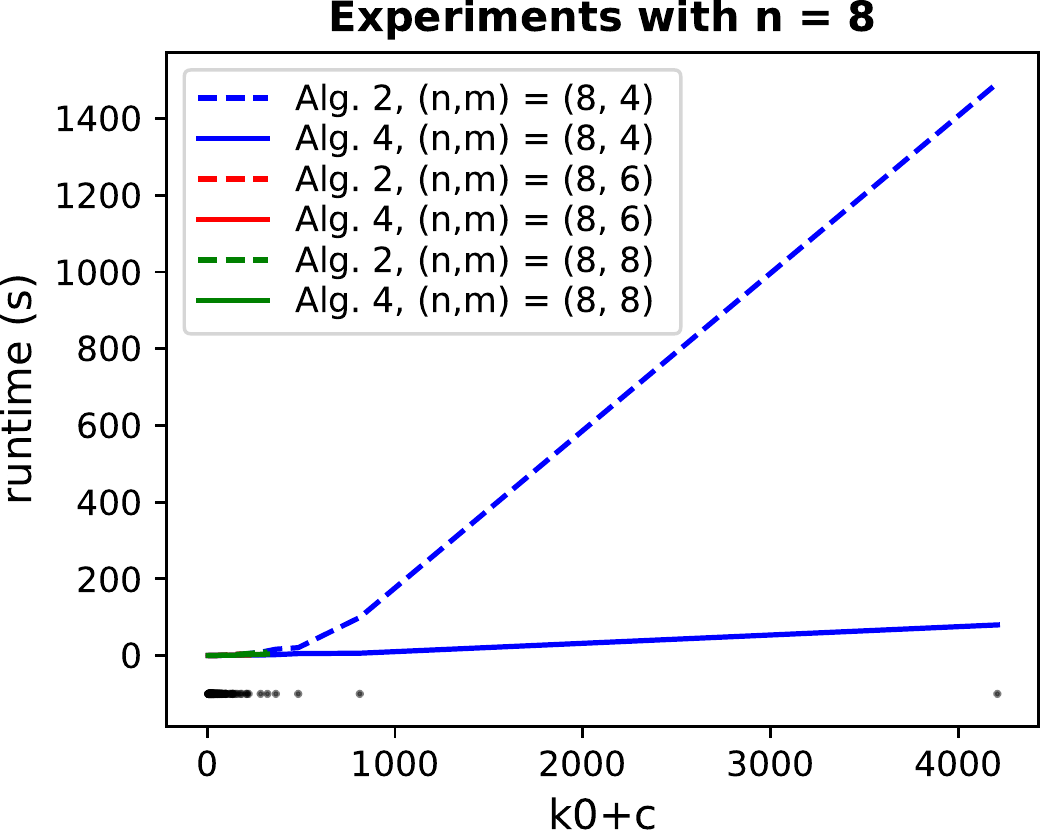}
\end{subfigure}
\end{figure}
\begin{figure}[!ht]
\begin{subfigure}[t]{.5\textwidth}
  \centering
  \includegraphics[scale=0.55]{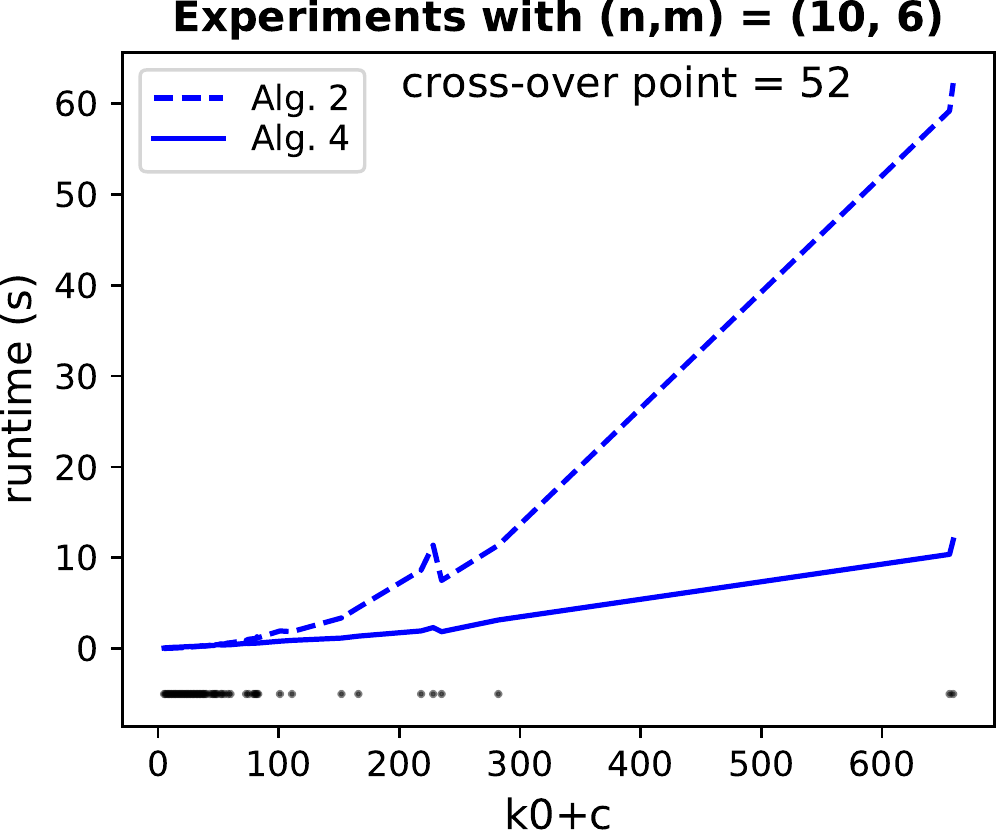}
\end{subfigure}
\begin{subfigure}[t]{.5\textwidth}
  \centering
  \includegraphics[scale=0.55]{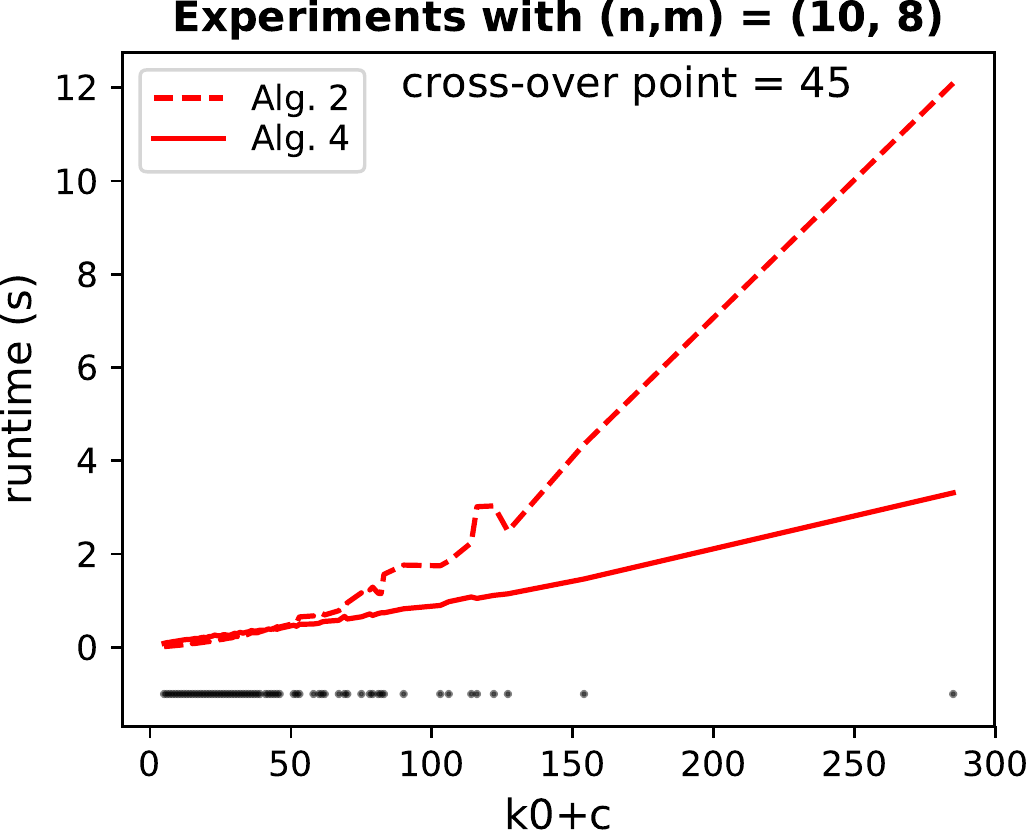}
\end{subfigure}
\vspace{0ex}\\
\begin{subfigure}[t]{.5\textwidth}
  \centering
  \includegraphics[scale=0.55]{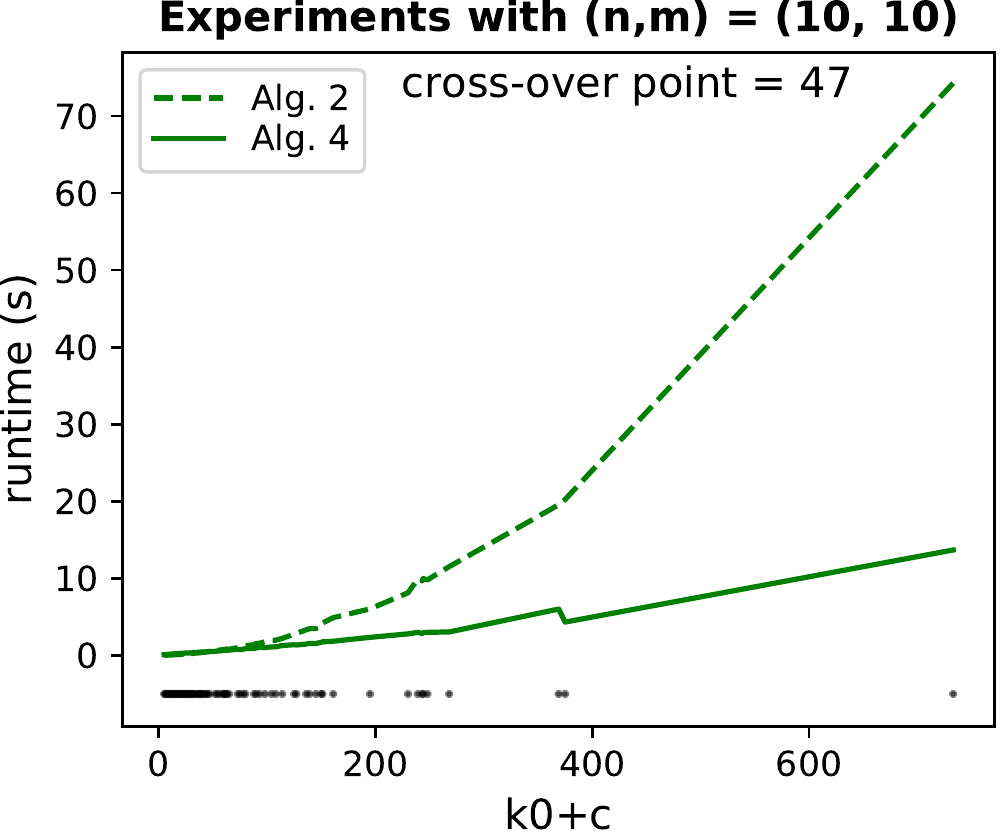}
\end{subfigure}
\begin{subfigure}[t]{.5\textwidth}
  \centering
  \includegraphics[scale=0.55]{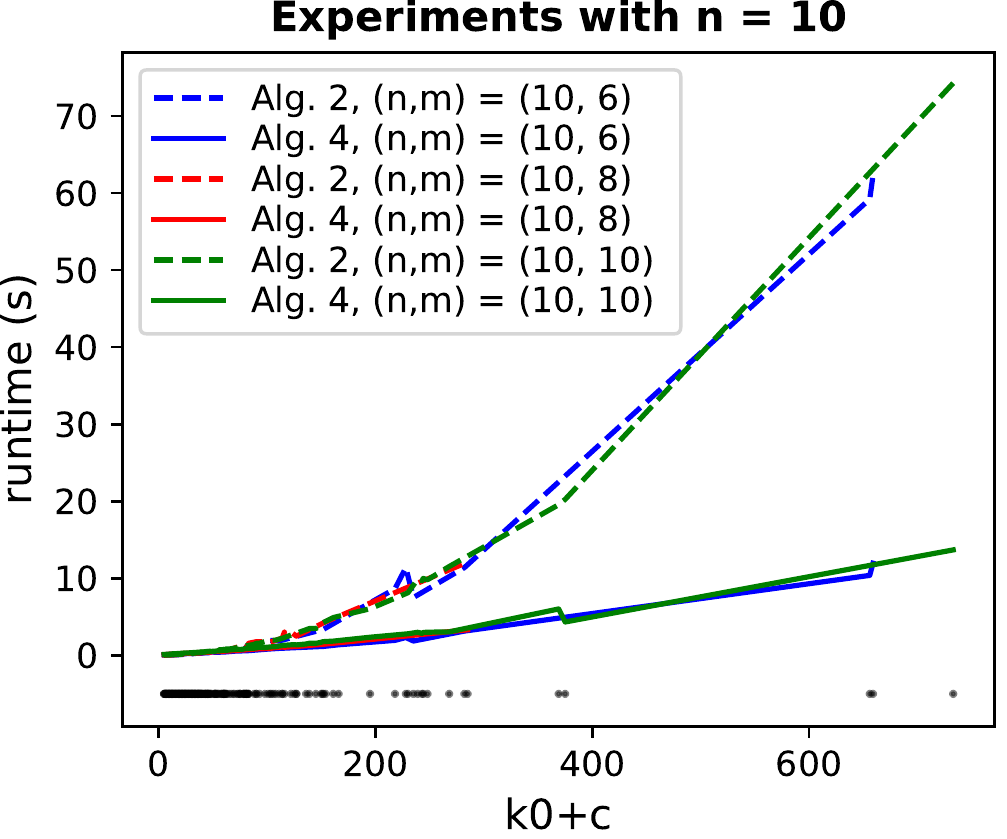}
\end{subfigure}
\vspace{0ex}\\
\begin{subfigure}[t]{.5\textwidth}
  \centering
  \includegraphics[scale=0.55]{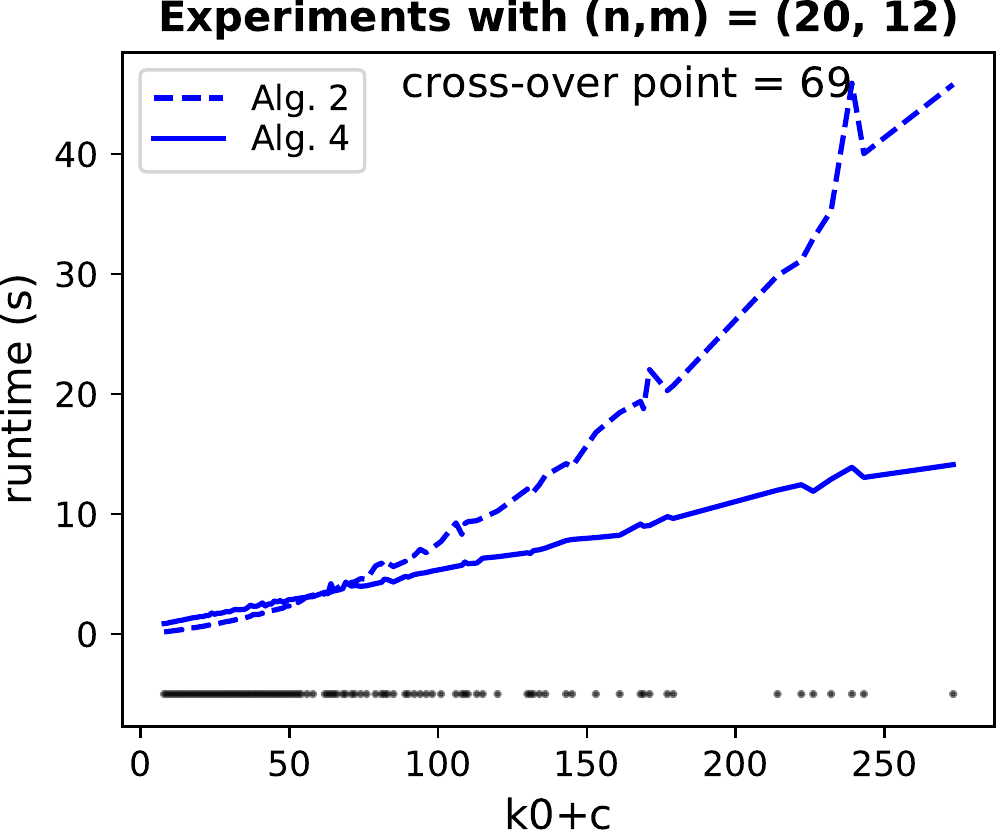}
\end{subfigure}
\begin{subfigure}[t]{.5\textwidth}
  \centering
  \includegraphics[scale=0.55]{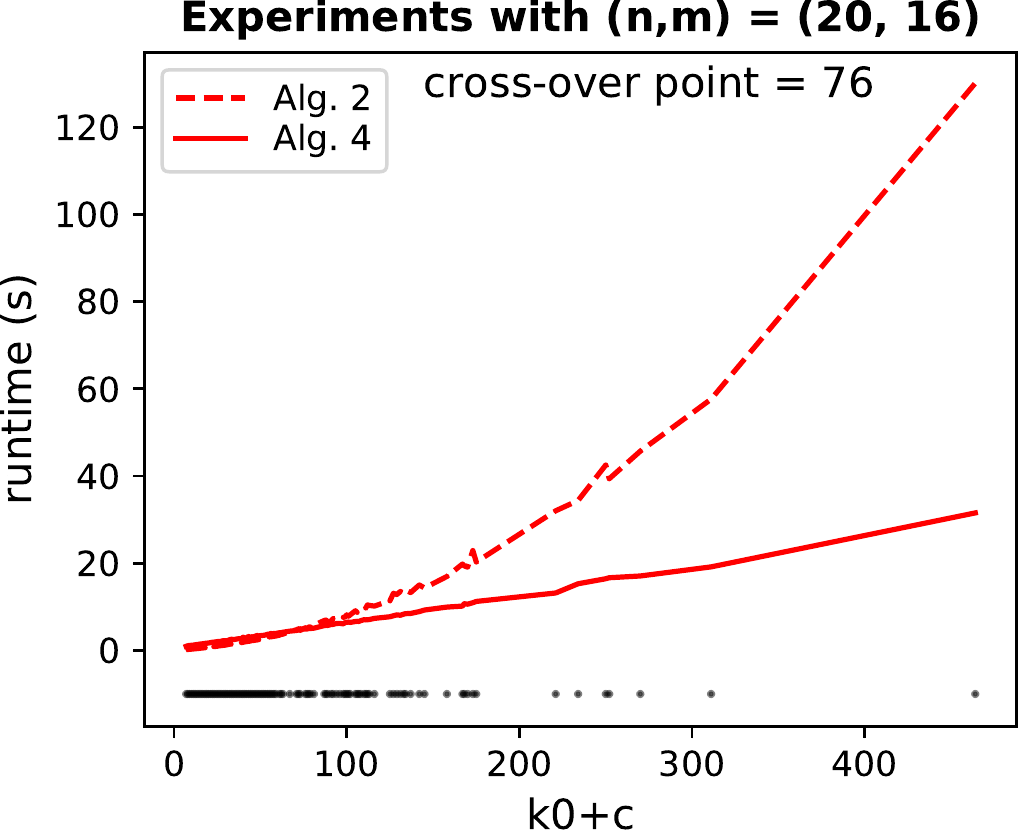}
\end{subfigure}
\vspace{0ex}\\
\begin{subfigure}[t]{.5\textwidth}
  \centering
  \includegraphics[scale=0.55]{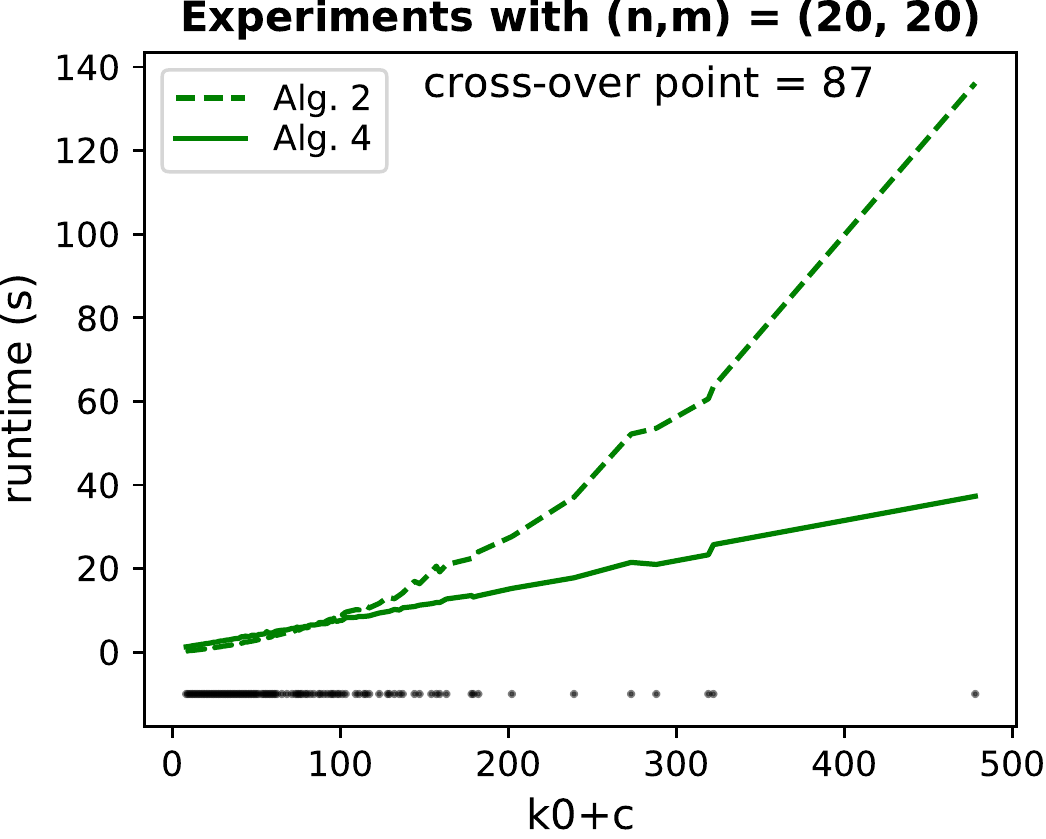}
\end{subfigure}
\begin{subfigure}[t]{.5\textwidth}
  \centering
  \includegraphics[scale=0.55]{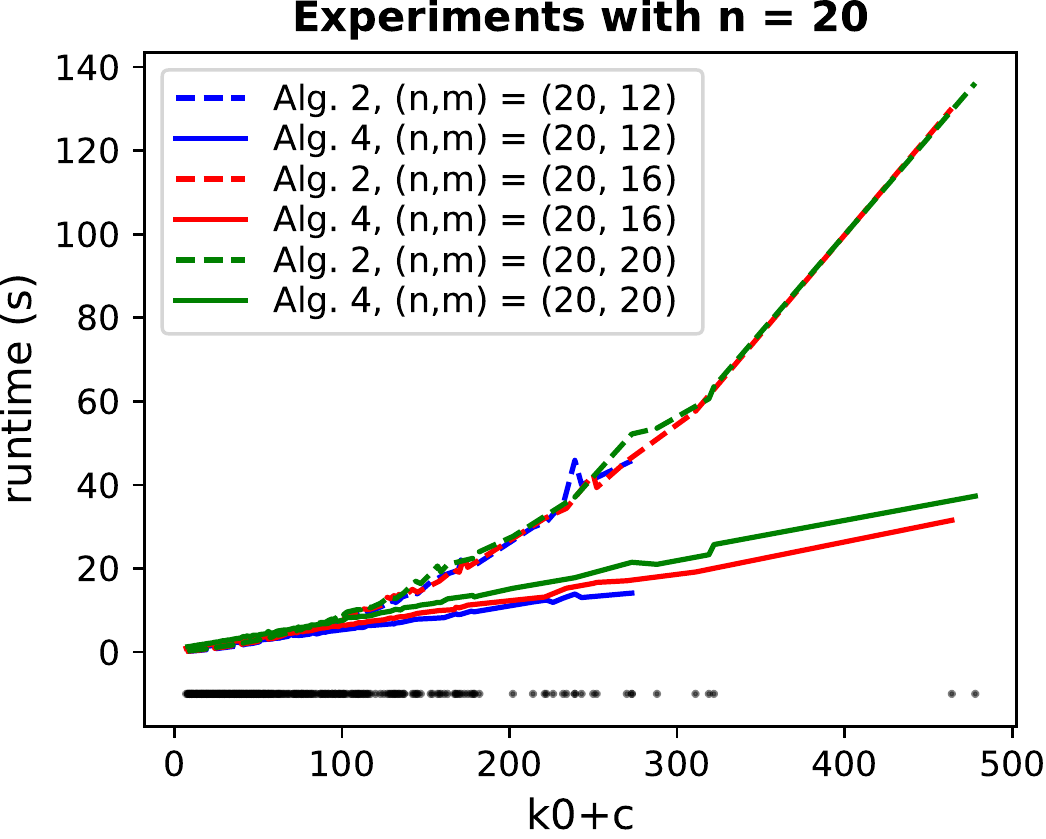}
\end{subfigure}
\end{figure}
\begin{figure}[!ht]
\begin{subfigure}[t]{.5\textwidth}
  \centering
  \includegraphics[scale=0.55]{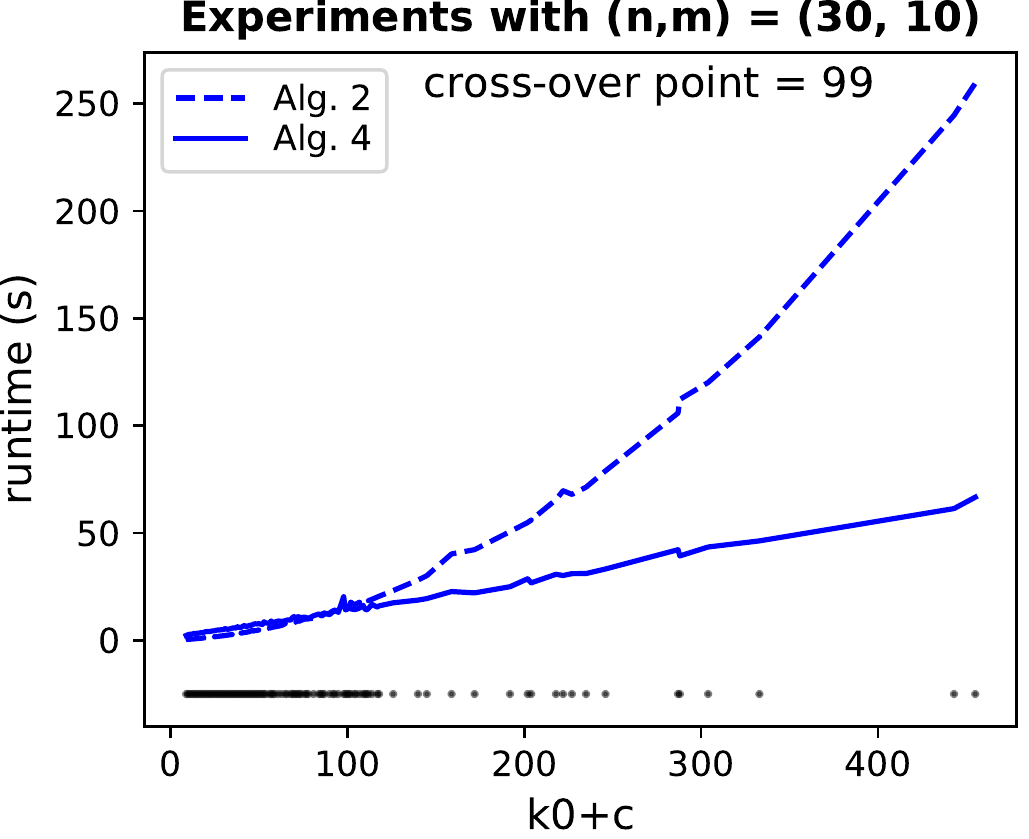}
\end{subfigure}
\begin{subfigure}[t]{.5\textwidth}
  \centering
  \includegraphics[scale=0.55]{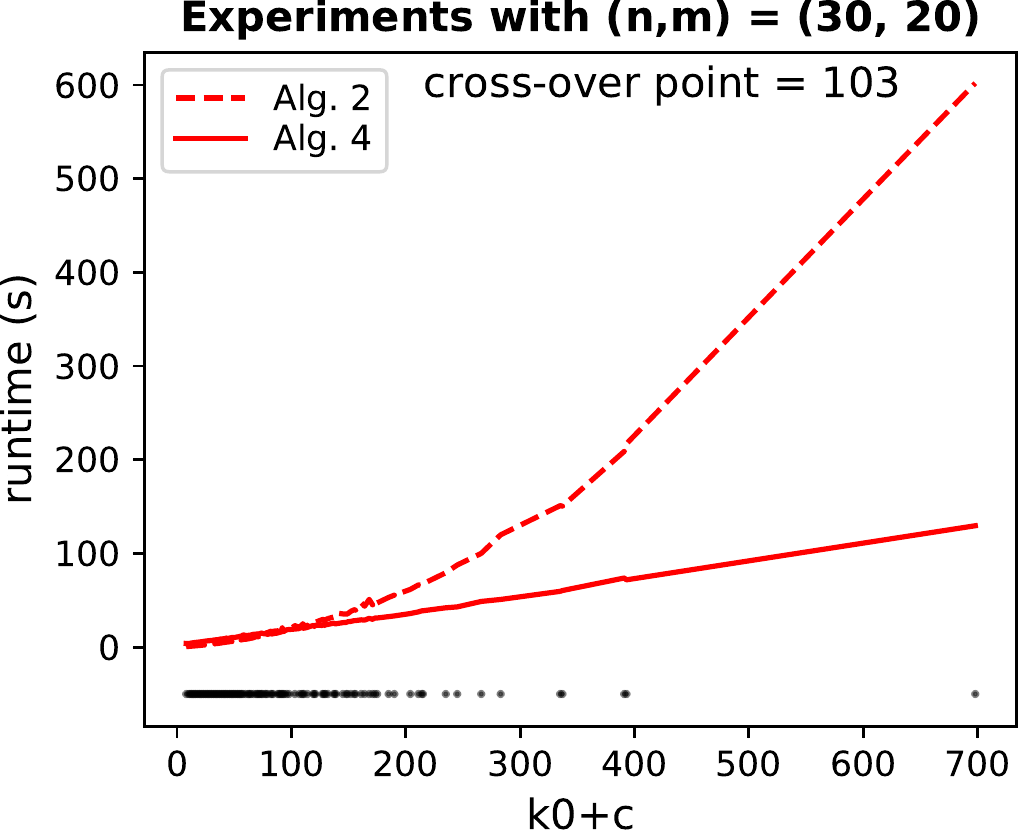}
\end{subfigure}
\vspace{0ex}\\
\begin{subfigure}[t]{.5\textwidth}
  \centering
  \includegraphics[scale=0.55]{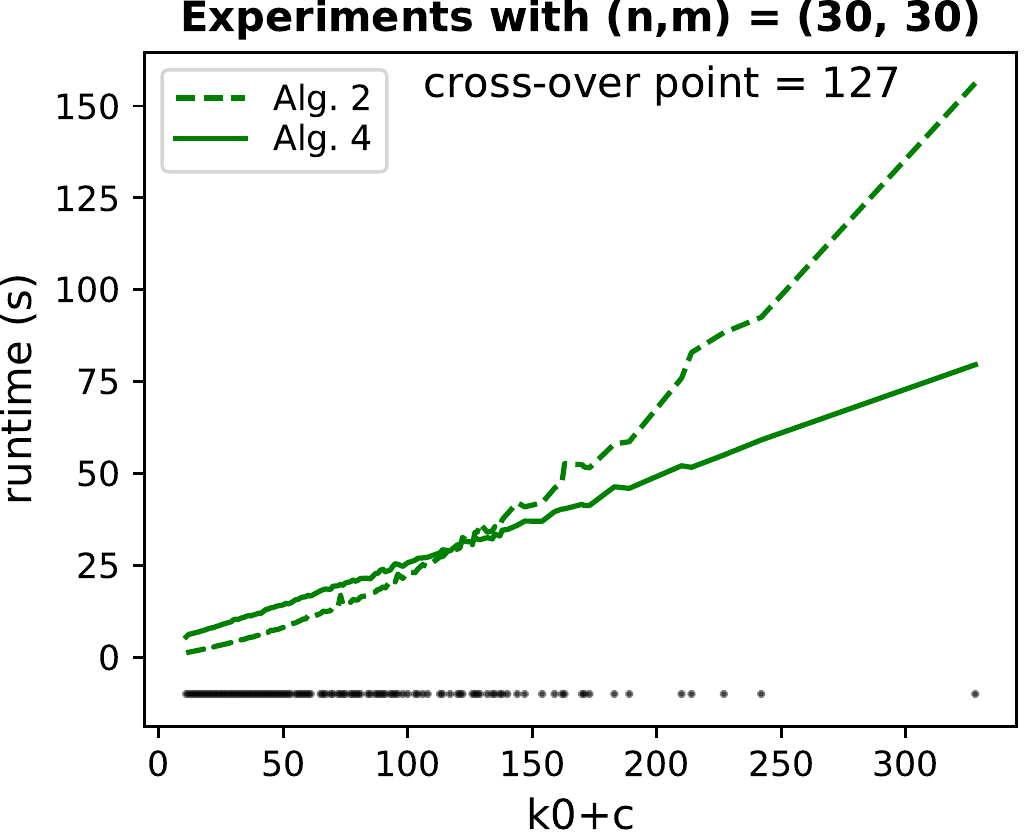}
\end{subfigure}
\begin{subfigure}[t]{.5\textwidth}
  \centering
  \includegraphics[scale=0.55]{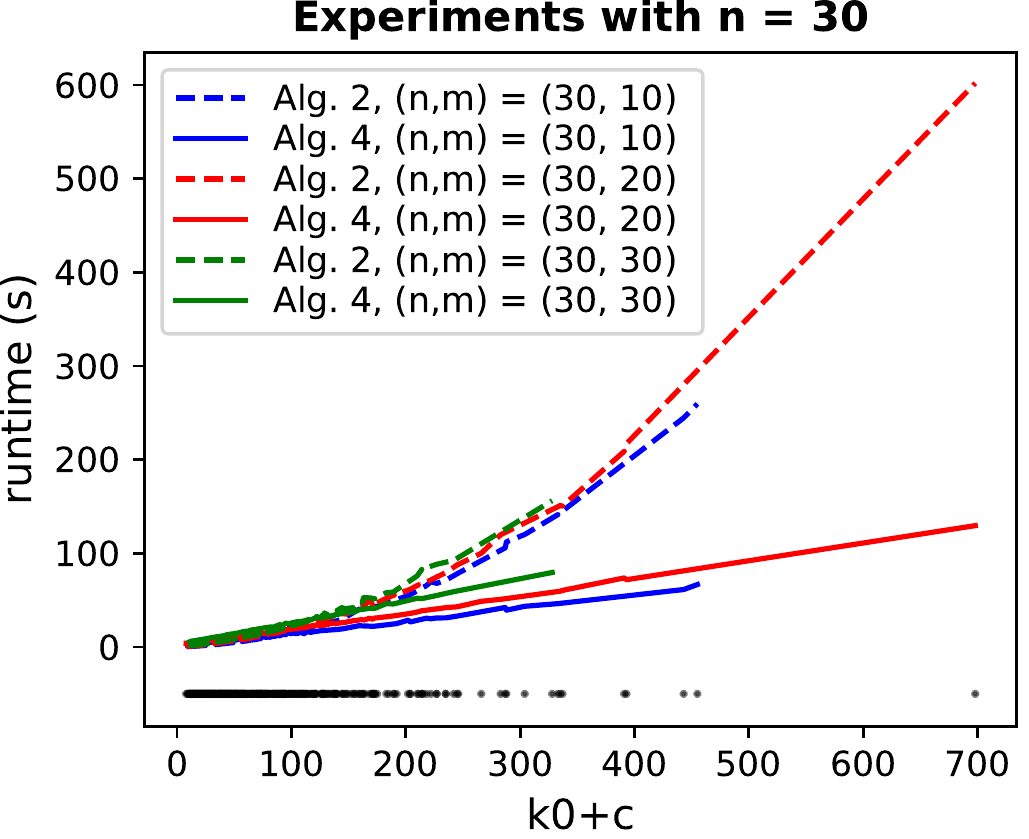}
\end{subfigure}
\vspace{0ex}\\
\begin{subfigure}[t]{.5\textwidth}
  \centering
  \includegraphics[scale=0.55]{plot40a.pdf}
\end{subfigure}
\begin{subfigure}[t]{.5\textwidth}
  \centering
  \includegraphics[scale=0.55]{plot40b.pdf}
\end{subfigure}
\vspace{0ex}\\
\begin{subfigure}[t]{.5\textwidth}
  \centering
  \includegraphics[scale=0.55]{plot40c.pdf}
\end{subfigure}
\begin{subfigure}[t]{.5\textwidth}
  \centering
  \includegraphics[scale=0.55]{plot40d.pdf}
\end{subfigure}
\end{figure}
\end{subappendices}
\end{document}